\documentclass[conference]{IEEEtran}
\IEEEoverridecommandlockouts
\usepackage{cite}
\usepackage{amsmath,amssymb,amsfonts}

\usepackage{graphicx}
\usepackage{textcomp}
\usepackage{xcolor}

\usepackage{times}
\usepackage{lipsum} 
\usepackage{epsfig}
\usepackage{graphicx}
\usepackage{amsmath}
\usepackage{amssymb}
\usepackage{algorithm}
\usepackage{algpseudocode}
\usepackage{booktabs}
\usepackage{makecell}
\usepackage{caption}
\usepackage{subcaption}
\usepackage{balance}

\usepackage{indentfirst}
\usepackage{multirow}
\usepackage{rotating}
\usepackage{array}
\usepackage{makecell}
\usepackage{tabularx}

\newcolumntype{C}[1]{>{\centering\arraybackslash}m{#1}}
\newcolumntype{Y}{>{\centering\arraybackslash}X}

\usepackage{bm}
\usepackage{amssymb}
\usepackage{pifont}

\usepackage{xcolor}
\usepackage{soul}
\usepackage{float}

\usepackage{amsmath,amsthm,amssymb}
\newtheorem{theorem}{Theorem}

\newtheorem{definition}{Definition}
\newtheorem{assumption}{Assumption}
\theoremstyle{remark}
\theoremstyle{claim}
\newtheorem*{remark}{Remark}

\renewcommand{\theequation}{\arabic{equation}}

\theoremstyle{remark}

\usepackage[caption=false, font=footnotesize, labelfont=bf, textfont=sf]{subfig}

\newcommand{\myparatight}[1]{\smallskip\noindent{\bf {#1}:}~}
\usepackage{xspace}

\newcommand{\algO}{\textsf{SRU}\xspace}

\newcommand{\algM}{\textsf{PRU}\xspace}
\newcommand{\algMns}{{\textsf{PRU}}}

\usepackage{cleveref} 
\crefname{axiom}{axiom}{axioms} 
\crefname{definition}{definition}{definitions}
\crefname{lemma}{lemma}{lemmata}

\def\BibTeX{{\rm B\kern-.05em{\sc i\kern-.025em b}\kern-.08em
    T\kern-.1667em\lower.7ex\hbox{E}\kern-.125emX}}

\begin{document}

\setlength{\textfloatsep}{4pt}
\setlength{\intextsep}{4pt}
\setlength{\floatsep}{4pt}

\title{Network Digital Untwinning: Towards Backward Optimization of Digital Twins}

\author{
		\IEEEauthorblockN{
            Zifan Zhang$\IEEEauthorrefmark{1}$,
            Dianwei Chen$\IEEEauthorrefmark{2}$,
            Anjun Gao$\IEEEauthorrefmark{3}$,
            Manhua Wang$\IEEEauthorrefmark{4}$, \\
            Mingzhe Chen$\IEEEauthorrefmark{5}$,
            Minghong Fang$\IEEEauthorrefmark{3}$,
            Xianfeng Yang$\IEEEauthorrefmark{2}$,
            Yuchen Liu$\IEEEauthorrefmark{1}$
 }
		\IEEEauthorblockA{
		$\IEEEauthorrefmark{1}$North Carolina State University, USA,
        $\IEEEauthorrefmark{2}$University of Maryland, USA, \\
		$\IEEEauthorrefmark{3}$University of Louisville, USA
        $\IEEEauthorrefmark{4}$University of Michigan, USA
        $\IEEEauthorrefmark{5}$University of Miami, USA
		}
}

\maketitle
\begin{abstract}
Network digital twins (NDTs) are transforming network management by offering precise virtual replicas of physical network systems. However, their reliance on diverse and sensitive data introduces significant challenges related to data management, regulatory compliance, and user privacy. In scenarios where selective data removal is necessary, such as device deactivation, network reconfiguration, or regulatory compliance, traditional approaches often fall short of preserving the integrity of the twin model.
To address this gap, we introduce a network digital untwinning framework that enables the targeted removal of deprecated NDT contributions while maintaining model integrity. Our approach comprises two complementary components: Single Request Untwinning (\algO) and Parallel Request Untwinning (\algM) mechanisms. \algO leverages connectivity metrics based on geographical proximity, data distribution, and network-level attributes to identify and remove the target NDT along with its propagating influence. This is achieved through an optimally selected rollback checkpoint augmented with injected Gaussian noise, followed by a precise remapping phase. \algM extends this mechanism to efficiently handle multiple removal requests by clustering NDTs with similar attributes and performing a coordinated rollback and untwinning schedule. We provide theoretical guarantees on model indistinguishability from scratch-built twins, and validate the framework through extensive experiments on real-world traffic data, demonstrating its effectiveness and operational efficiency.

\end{abstract}



\section{Introduction}

In recent years, network digital twins (NDTs) have emerged as a key enabler for next-generation (nextG) networks by providing high-fidelity virtual replicas of physical infrastructure~\cite{wu2021digital,zhang2025synergizing}. 
These digital representations support precise network planning, resource optimization, and performance evaluation prior to deployment, reducing risks and improving reliability. 
NDTs are also central to adaptive decision-making in dynamic settings such as cellular networks~\cite{he2022resource}, IoT~\cite{elayan2021digital}, and autonomous vehicles~\cite{wang2024smart, li2025towards, pegurri2025van3twin}. 
Through continuous bi-directional synchronization, NDTs enable real-time diagnostics, efficient resource use, and proactive maintenance~\cite{khan2022digital, tao2024wireless}. 
They also support agile cache management~\cite{zhang2024digital}, and leverage analytics to generate corner-case scenarios for testing safety-critical systems~\cite{villa2024colosseum}, i.e., serving as virtual testbeds that enable safe validation of new configurations, protocols, and technologies without impacting real-world infrastructure.

\myparatight{Motivation}
Despite its visible benefits, the adoption of NDTs in nextG networks introduces complex challenges related to data management, security, and regulatory compliance~\cite{zhang2025two, wang2024secure, jabeen2025securing, wang2025rsaka, li2024secr}. 
The twinning process depends on large-scale, distributed, and often sensitive datasets to simulate and optimize network behavior, including aspects such as traffic flows, device interactions, service demands, and dynamic changes in a network topology~\cite{zhang2025digitalmag}. 
This dependency exposes them to risks concerning data quality and privacy, particularly in collaborative, multi-user environments. Critical scenarios arise when user or device data must be selectively removed without compromising the NDT’s integrity, such as when users disconnect, devices deactivate, or network configurations evolve. 
Failing to address these scenarios can result in outdated, low-quality, or even adversarial data remaining in the system, degrading model performance and exposing vulnerabilities~\cite{zhang2024securing}.
Moreover, strict regulatory frameworks like General Data Protection Regulation (GDPR) and California Consumer Privacy Act (CCPA) require mechanisms that enable users to request data erasure, further emphasizing the need for a \textit{retrievable} NDT architecture, as specified in the ITU standard~\cite{Y3090,Y3091}. 
This calls for robust \textit{untwinning} capabilities, the process of selectively reversing or backward optimizing NDT models to eliminate unregulated data while preserving their reliability and functionality. 
To align with the capability levels defined in these standards, the NDT lifecycle must evolve from a unidirectional mapping process to a bidirectional framework. This requires embedding data erasure and model rollback mechanisms directly into the core network control loop.
Despite its significance, such backward optimization remains underexplored in current NDT research, which predominantly focuses on forward twinning and model construction. This work bridges this critical gap by introducing backward twinning mechanisms to ensure privacy compliance, uphold model quality, and enable effective responsiveness to dynamic operational and regulatory demands.

\myparatight{Challenges and Gaps}
Modeling the network digital untwinning process presents several unseen challenges. 
It is not merely a \textit{data-centric} unlearning problem, but a \textit{scenario-wide} backtracking procedure. First, NDTs are not independent clients, as their attributes are coupled through interference constraints and traffic-flow conservation laws, meaning that removing a single twin can propagate ripple effects across the entire network. 
Second, NDTs must maintain synchronization with their physical counterparts; full remapping would violate real-time constraints, making operational efficiency in model untwinning essential. Therefore, selectively removing network entities and their associated data contributions requires a careful balance between preserving twin accuracy and minimizing computational overhead.
From the data-level standpoint, privacy preservation emerges as another key challenge\footnote{While our primary focus is not on privacy protection in NDT systems, the proposed method inherently mitigates risks of data and model leakage.}. 
During the untwinning process, there is a heightened risk of inadvertently exposing sensitive information linked to the removed nodes within the updated NDT, particularly in collaborative or adversarial wireless environments. 
Existing machine unlearning methods~\cite{Fraboni2023sifu,jiang2024towards,10189868, 10234397,liu2020federaser} are ill-suited for direct adaptation to network digital untwinning, as detailedly evaluated in later Section~\ref{sec:exp}.
Specifically, retraining-based or exact unlearning often needs full or near-full re-optimization, which is too slow for the strict synchronization cadence of NDTs. 
Other unlearning methods rely on locality and small-perturbation assumptions, but NDT telemetry is globally coupled and dynamic, so removing one entity reshapes the effective data distribution and network constraints for others, violating both the replay-locality assumption of sharding-based methods and the smooth-change assumption of gradient or influence-based approximations.
Federated unlearning also often assumes clients are independent, but NDTs break this assumption because their states are linked through shared spectrum, topology, and traffic dynamics.
Additionally, most unlearning approaches are designed for classification tasks and lack the versatility to address regression-based models—central to NDTs—that always simulate continuous spatio-temporal variables and are tightly coupled to environmental dynamics.

\myparatight{Main Contribution}
To address the aforementioned limitations, we propose, for the first time, a network digital untwinning scheme that strategically removes joint data-scenario entities from constructed NDTs. The framework comprises two key components: single-request untwinning (\algO) and parallel-request untwinning (\algM). 
Given the challenges of mapping a one-size-fits-all digital twin, both \algO and \algM consider the interrelationships among a set of distributed local NDTs coordinated by a global twin as well-studied in~\cite{zhang2025digitalmag,zhang2024mapping,chen2024distributed, wu2021digital}.
Specifically, \algO targets backward optimization for individual untwinning requests, while \algM is tailored to manage multiple concurrent requests. 
When an untwinning request is issued such as node deactivation under data retention policies, \algO identifies and removes the target NDT’s \textit{influence} by computing dependency scores based on twinning region proximity, data similarity, and interconnectivity of mapped nodes. It then selects an optimal rollback point, injects Gaussian noise to obscure residual patterns, and triggers a fast remapping phase to restore model integrity.

In large-scale networks with interconnected NDTs~\cite{lin20236g}, \algM processes simultaneous untwinning requests by clustering NDTs based on geographical and communication attributes. It then selects rollback points using sensitivity criteria, performs parallel rollbacks with noise injection to erase partial contributions, and re-aggregates models to maintain a consistent global twin with minimal degradation.
\algM extends beyond targeting a single NDT by simultaneously untwinning a set of strongly connected twins based on network-level attributes. This approach is driven by the observation that closely linked twins often share redundant and correlated knowledge, making isolated removal ineffective and leaving residual influence in neighboring twins.

The main contributions can be summarized as follows:
\begin{list}{\labelitemi}{\leftmargin=1em \itemindent=-0.08em \itemsep=.2em}
    \item We propose a network digital untwinning framework, comprising \algO and \algMns, as the first solution specifically designed for NDTs to address the challenge of efficient model reversal with accurate contribution removal and adaptive checkpoint saving. The framework supports untwinning requests in both singular and parallel forms, with configurable strategies tailored to different operational contexts.

    \item We provide a theoretical guarantee on both \algO and \algMns, illustrating the indistinguishability between our updated twin models and map-from-scratch twin models.

    \item We conduct extensive experiments within a real-world vehicular system assisted by NDTs, including tasks such as traffic flow and vehicle speed forecasting. The evaluation results demonstrate the effectiveness and efficiency of our proposed framework in untwinning the corresponding NDTs in response to scalable requests.

\end{list}	

\section{Background and Preliminaries}
\label{sec:prelim}

Consider a scenario with multiple connected vehicles. Traditional federated learning (FL) systems treat each vehicle or deployed sensors as an independent client, training locally and sending updates to a roadside unit (RSU)-based server. However, in NDT systems, clients are interdependent—sharing spectrum, RSUs, and exhibiting correlated traffic. To address this, we (i) cluster clients with similar conditions (e.g., same lane or intersection), (ii) conduct twinning/untwinning within clusters to prevent indirect data leakage, and (iii) roll back only affected clusters. This \textbf{clustering-on-rollback strategy}, absent in conventional FL and unlearning studies, enables efficient and consistent untwinning in interdependent environments. As the first untwinning framework, our workflow will align with the established twinning process for consistency and clarity.

\myparatight{Network digital twinning}
In an NDT-assisted vehicular network~\cite{hakiri2024comprehensive, ding2025digital}, the local task-oriented NDTs are deployed on sensors along the road, while an edge computing server hosts a global twin deployed on the RSU for central coordination.
Note that $\bm{w}_n^t$ represents the local twin model $n \in [N]$ at mapping round $t \in [T]$, while $\bm{w}^t$ refers to the global twin model hosted on the roadside server during round $t$.
The objective is to minimize the traffic prediction error across $N$ NDTs $\bm{w}_n^t$ based on their sensory data, indexed by $n$. 
This problem can be framed as a \textbf{\textit{forward optimization problem}} to construct an optimal global twin model $\bm{w}^*$ stored on a roadside server, similar to the prior works in~\cite{wang2022mobility, zhang2021digital}:
\begin{align}
\label{all_obj}
w^{*,t} = \arg\min_{w} \frac{1}{N} \sum_{n=1}^{N} G(g(a_n^t, w), b_n^t),
\end{align}
where $G$ represents the loss function, defined as $G(g(a_n^t,$  $\bm{w}^t), b_n^t) = \left\| g(a_n^t, \bm{w}^t) - b_n^t \right\|$. 
Here, $g(a_n^t, \bm{w}^t)$ denotes the prediction made by the global twin model $\bm{w}^t$ for input $a_n^t$, and $b_n^t$ is the corresponding true label or target. 
This optimization task can be executed in a distributed fashion, as specified in the ITU standards~\cite{Y3090,Y3091}.
The forward twinning process includes four main steps: 
\begin{list}{\labelitemi}{\leftmargin=1em \itemindent=-0.08em \itemsep=.2em}
    \item \textbf{Step I (Local NDT mapping).}  
    Each NDT $n$ collects its time-series data $D_n$ from the local sensor and utilizes the initial global twin model parameters $\bm{w}^{t-1}$ to perform local mapping. This local mapping adapts the previous global twin model with local data, yielding an updated local twin model $\bm{w}_n^{t}$. Each NDT then transmits this locally mapped twin model back to the RSU server. 
    \item \textbf{Step II (Global NDT aggregation).}  
    Once all updated local twin models $\bm{w}_n^{t}$ have been received, the server aggregates them using a predefined aggregation rule ($\text{AGGR}$) to produce an updated global twin model $\bm{w}^{t}$.
    A commonly-used aggregation rule is the Federated Averaging (FedAvg)~\cite{mcmahan2017communication}, where RSU server computes an average of local twin models, i.e.,
    \(
    \text{AGGR}\{\bm{w}_1^{t}, \bm{w}_2^{t}, \ldots, \bm{w}_n^{t}\} = \frac{1}{N} \sum\limits_{n=1}^N \bm{w}_n^{t}.
    \)
    The global twin model $\bm{w}^{t}$ is then updated.
    \item \textbf{Step III (NDT synchronization).}  
    At the end of each mapping round, the roadside server broadcasts the current global twin model $\bm{w}^t$ to all participating NDTs. 
    %
    \item \textbf{Step IV (Scenario reconstruction and generation).}  
    Throughout $T$ twinning rounds, the global twin model $\bm{w}^t$ is iteratively refined to capture patterns from the heterogeneous data and scenario information while preserving the privacy of each NDT's local dataset. 
    
\end{list}

\myparatight{Network Digital Untwinning}
After the twinning process is completed, the network system may receive requests to remove the influence of deprecated users or sensors from a particular NDT. 
To ensure compliance and adaptability, an untwinning process must be performed. 
Consequently, a \textbf{\textit{backward optimization problem}} is formulated to retrieve the optimal global twin model $\bm{w}^*$, as detailed in Sec.~\ref{sec:formulation}.
This procedure essentially reverts $\bm{w}^t$ to a safe checkpoint and remaps it to eliminate any effects from target NDTs $\bm{w}_n^{t}$, thereby achieving a positive forgetting, as defined as follows: 

\begin{list}{\labelitemi}{\leftmargin=1em \itemindent=-0.08em \itemsep=.2em}
    \item \textbf{Step I (Removal decision).}  
    Each NDT $n$ either (i) autonomously detects deprecated data or malicious entities that should be removed from the twinning process, or (ii) receives a user request $u \in [U]$ to untwin its private data or related NDTs. If $[U] \neq \emptyset$, the system proceeds with the following untwinning steps.
    \item \textbf{Step II (Checkpoint rollback).}  
    Let $K$ be the number of roll-back rounds, determined by the sensitivity of target NDTs in $[U]$. The global twin model reverts $\bm{w}^{t}$ to a prior checkpoint $\bm{w}^{t-K}$ with minimal effects from target NDTs.

    \item \textbf{Step III (NDT optimization).}  
    Each remaining NDT resumes the forward twinning for $K$ rounds, or until convergence, using only the updated dataset. 
\end{list}

\section{Context-Aware Network Digital Untwinning}
\label{sec:method}

In Sec.~\ref{sec:prelim}, we presented the generic process of network digital twinning and introduced the basic concept of untwinning in distributed network settings. Building on this foundation, this section formally defines the problem of network digital untwinning. We then introduce \algO and \algM, which enable single-instance and parallel network digital untwinning processes, respectively, depending on the current conditions of the unlearning request (as detailed in Secs.~\ref{sec:algorithm_single} and~\ref{sec:algorithm_parallel}).

\subsection{Problem Formulation}
\label{sec:formulation}

In the network digital untwinning process, the server receives a set of unlearning requests, denoted as $\mathcal{U}$, which identifies the target local NDTs that must be removed to eliminate their influence on the global twin model $\bm{w}^T$.
The goal is to \textit{remap} (update) the global twin to reflect only the remaining data contexts, minimizing communication and computational costs compared to a full reconstruction (retraining from scratch).

However, obtaining such a global NDT cannot be achieved through simple aggregation adjustments (e.g., subtracting weights), as the local twin models are inherently coupled and require dynamic coordination. Therefore, remapping and fine-tuning techniques are essential. We formulate this problem as constructing an untwinned global NDT that is statistically indistinguishable from one generated through a full map-from-scratch process.

\begin{definition}[$(\epsilon, \beta)$-Indistinguishability: Unlearning Guarantee]
Let $X$ and $Y$ be random variables over a domain $\mathcal{R}$ (representing the distribution of trained models). For any subset $S \subseteq \mathcal{R}$, $X$ and $Y$ are $(\epsilon, \beta)$-indistinguishable if:
\begin{align}
    \Pr[X \in S] &\leq e^\epsilon \cdot \Pr[Y \in S] + \beta, \\
    \Pr[Y \in S] &\leq e^\epsilon \cdot \Pr[X \in S] + \beta.
\end{align}
\end{definition}

\myparatight{Main objective}
The objective is to efficiently remap a global NDT that is $(\epsilon, \beta)$-indistinguishable from a reference model obtained through a full map-from-scratch process on the retained dataset (excluding $\mathcal{U}$).
Formally, let $\bm{w}_{\text{untwinned}}$ be the model produced by our method and $\bm{w}_{\text{retrain}}$ be the model trained from scratch without $\mathcal{U}$. We require these two random variables to satisfy Definition 1.
This ensures that any observer is unable to distinguish whether the new global twin model was incrementally untwinned or entirely reconstructed, thereby guaranteeing the effective removal of the target data while significantly reducing system overhead.

\subsection{Single-Request Untwinning (\algO)}
\label{sec:algorithm_single}


\begin{figure}[!h]
    \centering
    \includegraphics[width=0.48\textwidth]{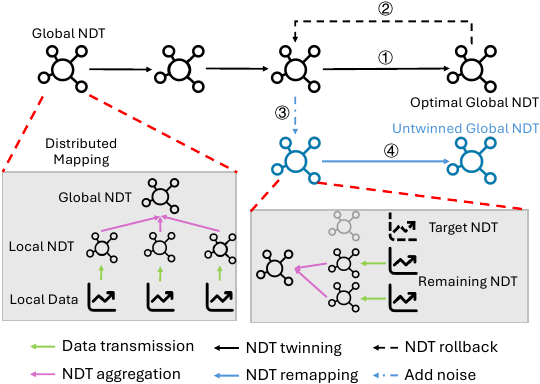}
    \caption{The workflow of \algO. We use neural network symbols to represent NDT models, but can be extended to any models. }
    \label{fig:single}
\end{figure}

\algO is designed to remove the impact of a specific target NDT $n_u$ from the global twin model $\bm{w}^T$. This process not only excludes direct contributions from $n_u$ but also mitigates \textit{indirect} effects propagated through other interrelated NDTs that share significant connectivity with $n_u$. \algO consists of three main steps.

\myparatight{1) Untwinning set identification}
Initially, the network system proceeds with standard digital twinning using local sensory data, as shown in Step 1 of Fig.~\ref{fig:single}.
Upon receiving an untwinning request for a target NDT $n_u$, we must identify correlated peers. Each NDT $n \in [N]$ is assigned an importance score $\mathcal{I}(n)$ relative to the target, based on a pre-computed connectivity matrix $\mathbf{C}$:
\(
\mathcal{I}(n) \leftarrow \mathbf{C}(n, n_u).
\)
The connectivity matrix $\mathbf{C}$ quantifies the coupling strength between twins. We construct this matrix based on an attribute sequence of the Base Stations (BSs), denoted as $\{ g, k, \delta, \tau \}$, representing geographical distance, backhaul link capacity, coverage overlap, and data distribution similarity, respectively. The correlation metric is defined as:
\begin{equation}
\Phi_{i,j} = \frac{\omega_g}{g_{i,j}} + \omega_k \cdot k_{i,j} + \omega_\eta \cdot \delta_{i,j} + \omega_\tau \cdot \tau_{i,j},
\end{equation}
where $\omega_{(\cdot)}$ represents the weighting coefficients for each attribute.
A higher importance score $\mathcal{I}(n)$ indicates that NDT $n$ is implicitly \textit{connected} to $n_u$, leading to potential data leakage during aggregation.
A user-specified untwinning threshold $\theta$ is applied to determine the \textit{untwinning set} $\mathcal{S}_u$ (Line~\ref{line:unlearned_set}, Algorithm~\ref{federated_untwinning_single}). The set $\mathcal{S}_u$ includes $n_u$ and any NDTs where $\mathcal{I}(n) > \theta$. This coordinated strategy effectively removes shared redundancies and correlated influences, enhancing privacy while retaining weakly connected twins to preserve model utility.

\myparatight{2) Checkpoint rollback}
In Line~\ref{line:remove_NDT}, Algorithm~\ref{federated_untwinning_single} excludes all NDTs in $\mathcal{S}_u$ from participation. Subsequently, we determine the optimal rollback depth $K$ based on a sensitivity metric (Lines~\ref{line:get_sigma}--\ref{line:get_K}).
Directly computing the exact influence of removed data for every round is computationally impractical. Instead, we approximate the contribution sensitivity. Let the instantaneous sensitivity at round $t$ be:
\begin{equation}
\Delta_t = \mathcal{I}(n_u) \cdot \bigl\lVert \bm{w}(\mathcal{D}) - \bm{w}(\mathcal{D}_{\setminus \mathcal{S}_u}) \bigr\rVert,
\end{equation}
where $\mathcal{I}(n_u)$ scales the impact based on the target's connectivity.
We then formulate the cumulative contribution $\phi(t)$ over the history:
\begin{equation}
\phi(t) = \sum_{\tau=0}^{t-1} (1 + \eta L)^{t-1-\tau} \cdot \Delta_{\tau},
\label{eq:contribution}
\end{equation}
where $\alpha_\tau$ represents a decay factor relative to the learning rate.
Using this accumulated influence, we derive a privacy-preserving rollback criterion $\gamma(t)$, calibrated against a Gaussian noise budget $\Omega$~\cite{dwork2014algorithmic}:
\begin{equation}
\label{eq:sensitivity}
    \gamma(t) = \frac{\Omega}{\epsilon \cdot \phi(t)}.
\end{equation}
The goal is to locate the most recent checkpoint $\bm{w}^{T-K}$ where the accumulated influence of the target set $\mathcal{S}_u$ is below a safety threshold $\gamma^*$. Thus, $K$ is determined by finding the minimal rollback distance:
\begin{equation}
K = T - \max \left\{ t \mid \gamma(t) \le \gamma^* \right\}.
\end{equation}

\myparatight{3) NDT perturbation and optimization}
After determining $K$, the system reverts the global model to $\bm{w}^{T-K}$. To guarantee the removal of residual information, we inject Gaussian noise (Line~\ref{line:rollback}):
\(
\widetilde{\bm{w}} = \bm{w}^{T-K} + \mathcal{N}(\mathbf{0}, \delta \mathbf{I}).
\)
This noise obfuscates trace patterns from $\mathcal{S}_u$ that may persist in the checkpoint.
Once $\widetilde{\bm{w}}$ is synchronized (Line~\ref{line:sync}), the system initiates remapping from round $t = T-K$ to $T$ (Step 4, Fig.~\ref{fig:single}).
In each round, remaining NDTs sample mini-batches, compute gradients $\bm{g}(\bm{w}_n^t)$, and send updates to the server. The server aggregates them via:
\begin{equation}
\label{eq:FedAvg}
\bm{w}^t \leftarrow \frac{1}{|N \setminus \mathcal{S}_u|} \sum_{n \in [N] \setminus \mathcal{S}_u} \bm{w}_n^t.
\end{equation}
By removing $\mathcal{S}_u$ and retraining from a clean checkpoint, \algO eliminates the \textbf{propagating influence} of the target while preserving the utility of the remaining participants. 

\begin{algorithm}[t]
\small
\caption{Single-Request Untwinning (\algO)}
\label{federated_untwinning_single}
\begin{algorithmic}[1]
    \Require NDTs $[N]$, connectivity matrix $\mathbf{C}$, threshold $\theta$, target $n_u$, global model $\bm{w}^T$, mapping rate $\eta$, safety threshold $\gamma^*$
    \Ensure Updated global twin model $\bm{w}^{T'}$

    \For {each NDT $n \in [N]$}
        \State Compute importance score $\mathcal{I}(n) \leftarrow \mathbf{C}(n, n_u)$
        \label{line:score_loop}
    \EndFor
    \State Identify untwinning set $\mathcal{S}_u \leftarrow \{ n \in [N] : \mathcal{I}(n) \geq \theta \} \cup \{n_u\}$
    \label{line:unlearned_set}
    \State Remove all $n \in \mathcal{S}_u$ from active set $[N]$
    \label{line:remove_NDT}
    
    \State \textbf{Calculate Rollback Step:}
    \State Compute influence curve $\gamma(t)$ for $t \in [0, T]$ via Eq.~(\ref{eq:sensitivity})
    \label{line:get_sigma}
    \State Determine rollback round: $t_{\text{safe}} = \max \{ t \mid \gamma(t) \leq \gamma^* \}$
    \State Set $K = T - t_{\text{safe}}$
    \label{line:get_K}
    
    \State \textbf{Rollback and Perturb:}
    \State $\tilde{\bm{w}} = \bm{w}^{T-K} + \mathcal{N}(0, \delta)$
    \label{line:rollback}
    \State Broadcast $\tilde{\bm{w}}$ to remaining NDTs
    \label{line:sync}
    
    \State \textbf{Retraining Phase:}
    \For {$t = T-K+1$ to $T$}
        \label{line:outer_for}
        \For {each remaining NDT $n \in [N]$}
            \label{line:inner_for}
            \State Sample batch from $D_n$
            \label{line:sample}
            \State Compute gradient $\bm{g}(\bm{w}_n^{t-1})$
            \label{line:grad}
            \State Update local twin $\bm{w}_n^t \leftarrow \bm{w}^{t-1} - \eta \bm{g}(\bm{w}_n^{t-1})$
            \label{line:update}
            \State Upload $\bm{w}_n^t$ to server
            \label{line:send_server}
        \EndFor
        \State $\bm{w}^t \leftarrow \text{AGGR}\bigl\{ \bm{w}_n^t \bigr\}$
        \label{global_update}
        \State Broadcast $\bm{w}^t$
        \label{line:sync2}
    \EndFor
\end{algorithmic}
\end{algorithm}

\subsection{Parallel-Request Untwinning (\algM)}
\label{sec:algorithm_parallel}

\begin{figure*}[!t]
    \centering
    \includegraphics[width=0.8\textwidth]{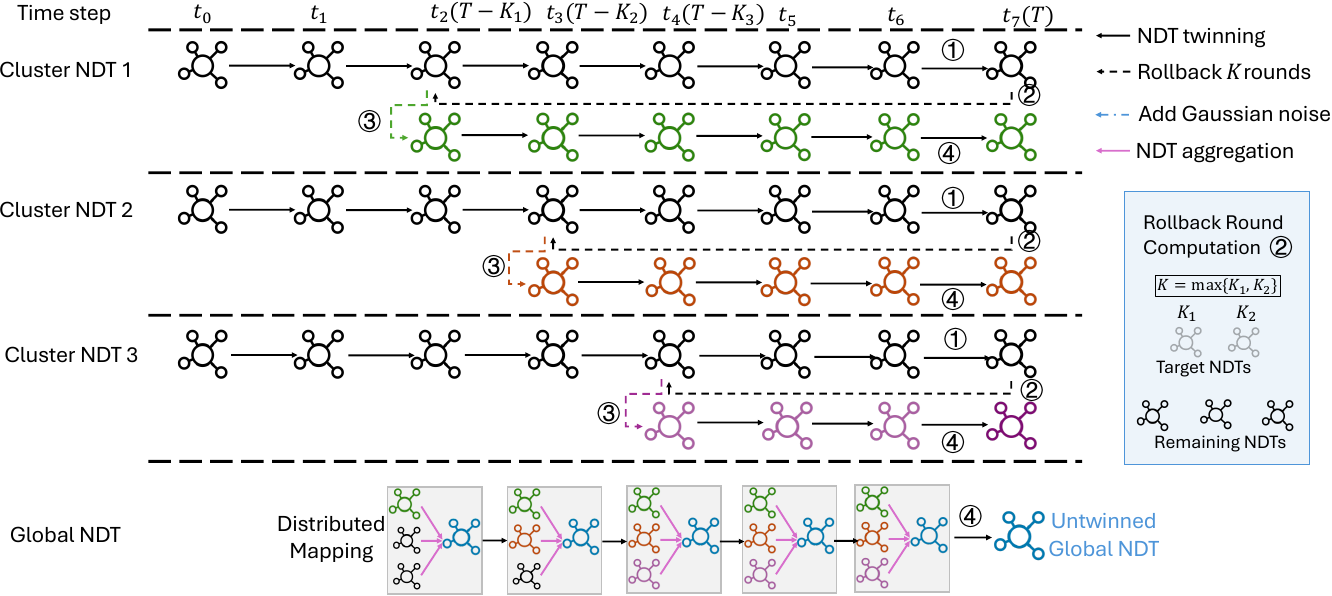}
    \caption{\algM consists of four main steps: (1) Parallel forward twinning; (2) Context-aware scheduling to determine cluster-specific rollbacks; (3) Cluster-level perturbation for obfuscation; (4) Staggered re-twinning until global convergence.}
    \label{fig:parallel}
    \vspace{-0.1in}
\end{figure*}

\algM addresses the scenario where the system receives multiple untwinning requests $\mathcal{U}$ from different local NDTs simultaneously, as shown in Fig.~\ref{fig:parallel} and Algorithm~\ref{federated_untwinning_multiple}.
Unlike \algO, which targets a single request, \algM adopts a clustering approach to parallelize the process. This improves efficiency by grouping NDTs and performing one rollback operation per cluster, managed by a context-aware scheduling mechanism.
\algM consists of three interrelated procedures:

\myparatight{1) NDT clustering}
Local NDTs are grouped into clusters $\mathcal{C} = \{C_1, \dots, C_M\}$. The clustering algorithm, adapted from~\cite{zhang2024mapping}, periodically groups NDTs with similar communication characteristics and network configurations.
Metrics include the geographical (not geological) distances between network entities and similarity of sensory data distributions.
The objective is to localize the influence of deprecated data to specific clusters, allowing system to roll back only the affected subgroups rather than entire network.
We employ an adaptive clustering strategy in which the server maintains a connectivity matrix and triggers re-clustering only when topology changes exceed a preset threshold. This ensures the overhead remains modest while accurately tracking system dynamics.

\myparatight{2) Context-aware scheduling}
For each cluster $C_a \in \mathcal{C}$, Algorithm~\ref{federated_untwinning_multiple} identifies the subset of relevant untwinning requests $\mathcal{U}_a = \mathcal{U} \cap C_a$ (Line~\ref{line:Ua_intersect}).
The system then calculates the required rollback depth for the cluster. Because a cluster may contain multiple target NDTs, the cluster must roll back far enough to satisfy the strictness requirement of the \textit{most sensitive} target in that group.
The cluster rollback round $K_a$ is determined by:
\begin{equation}
K_a = \max_{n \in \mathcal{U}_a} \left( T - \max \left\{ t \mid \gamma_n(t) \le \gamma^* \right\} \right).
\end{equation}
This ensures that by reverting $K_a$ rounds, the cluster model $\bm{w}_{C_a}$ reaches a checkpoint safe for all targets within it.

\myparatight{3) Cluster-level perturbation and optimization}
To obscure any residual updates from the removed NDTs, the rolled-back cluster model $\bm{w}^{T-K_a}$ is perturbed with Gaussian noise (Line~\ref{line:rollback_multiple}):
\begin{equation}
\widetilde{\bm{w}}_a = \bm{w}_a^{T-K_a} + \mathcal{N}(\mathbf{0}, \delta \mathbf{I}).
\end{equation}
This perturbed state serves as the consistent starting point for the remapping phase.
A global remapping phase is then executed from $t_{start} = T - \max_a(K_a)$ to $T$.
During this phase, clusters rejoin the training process in a staggered manner based on their rollback depth:
\begin{enumerate}
    \item \textbf{Active State ($t \ge T - K_a$):} If the current round $t$ is within the cluster's retraining window, the cluster proceeds with forward twinning. Each NDT $n \in C_a$ computes local updates, which are aggregated to form the updated cluster model $\bm{w}_a^t$.
    \item \textbf{Waiting State ($t < T - K_a$):} If the cluster has not yet reached its rollback point (i.e., it required less rollback than others), its model remains frozen at the widely synchronized state until the timeline catches up.
\end{enumerate}
For example, as shown in Fig.~\ref{fig:parallel}, Cluster 1 may have a large $K_1$ and begin retraining immediately, while Cluster 2 (with a small $K_2$) waits and joins the process later. This strictly preserves the accuracy of unaffected or less-affected clusters by preventing unnecessary rollback.
Finally, the server aggregates cluster-level models to produce the global twin $\bm{w}^T$.

\begin{algorithm}[t]
\small
\caption{Parallel-Request Untwinning (\algM)}
\label{federated_untwinning_multiple}
\begin{algorithmic}[1]
    \Require NDTs $[N]$, connectivity matrix $\mathbf{C}$, requests $\mathcal{U}$, global model $\bm{w}^T$, mapping rate $\eta$, clusters $\mathcal{C}$
    \Ensure Updated global twin model $\bm{w}^{T'}$

    \State Initialize max rollback $K_{\max} \leftarrow 0$
    
    \For {each cluster $C_a \in \mathcal{C}$} \label{line:for_cluster_start}
        \State Identify requests in cluster: $\mathcal{U}_a \leftarrow \mathcal{U} \cap C_a$
        \label{line:Ua_intersect}
        \If {$\mathcal{U}_a \neq \emptyset$}
            \State Calculate safety threshold $\gamma^*$ via Eq.~(\ref{eq:sensitivity})
            \label{line:get_sigma_multiple}
            \State Determine cluster rollback $K_a$ satisfying all $n \in \mathcal{U}_a$
            \label{line:get_Ka}
            \State Update max rollback: $K_{\max} \leftarrow \max(K_{\max}, K_a)$
            
            \State Rollback \& Perturb: $\tilde{\bm{w}}_a = \bm{w}^{T-K_a} + \mathcal{N}(0, \delta)$
            \label{line:rollback_multiple}
            \State Broadcast $\tilde{\bm{w}}_a$ to cluster members
            \label{line:sync_multiple}
        \Else
            \State $K_a \leftarrow 0$
        \EndIf
    \EndFor
    \label{line:for_cluster_end}

    \State \textbf{Staggered Retraining Phase:}
    \For{$t = T - K_{\max}$ to $T$} \label{line:outer_for_multiple}
        \For{each cluster $C_a \in \mathcal{C}$} \label{line:inner_for_cluster}
            
            \If{$t \ge T - K_a$} \label{line:if_t_greater}
                \For {each NDT $n \in C_a$} \label{line:for_each_NDT_na}
                    \State Sample batch from $D_n$
                    \State Compute local update $\bm{w}_n^t \leftarrow \bm{w}_n^{t-1} - \eta \bm{g}(\bm{w}_n^{t-1})$
                    \label{line:update_n}
                \EndFor
                \State Aggregate cluster model: \(\bm{w}_a^t \leftarrow \text{AGGR}\bigl\{ \bm{w}_n^t : n \in C_a \bigr\}\)
                \label{line:aggregate_cluster_a}
            \Else 
                \State Hold state: \(\bm{w}_a^t \leftarrow \bm{w}_a^{t-1}\)
                \label{line:no_update_cluster_a}
            \EndIf
            
        \EndFor \label{line:end_inner_for_cluster}
        
        \State Global Aggregation: \(\bm{w}^t \leftarrow \text{AGGR}\Bigl\{ \bm{w}_a^t : C_a \in \mathcal{C} \Bigr\}\)
        \label{line:agg_across_clusters}
        \State Synchronize $\bm{w}^t$ to all clusters
    \EndFor \label{line:end_outer_for_multiple}
\end{algorithmic}
\end{algorithm}

\subsection{Adaptive Checkpoint Saving and Complexity Analysis}
\label{sec:adaptive-ckpt}

Both \algO and \algM rely on rolling back the global twin model to a historical checkpoint (e.g., $\bm{w}^{T-K}$) before injecting Gaussian noise and remapping. A naive strategy of persisting the model at every round $t$ guarantees exact rollback but incurs a storage cost of $\mathcal{O}(T\cdot d)$, which is prohibitive for highly evolved NDTs. Note that $d$ represents the dimensionality of twin model parameters.
To address this, we design an adaptive topology-aware checkpointing (ATAC) mechanism. ATAC optimizes the \textit{storage-vs-compute} trade-off by maintaining a sparse, non-uniform timeline of checkpoints, prioritizing moments of high model drift and topological structural change.

\myparatight{1) Checkpoint utility analysis}
We identify two key signals that necessitate a snapshot, which are \textit{model instability} and \textit{topology shifts}. At the end of round $t$, the server monitors the drift in the global twin model $\bm{w}^t$ and the updated connectivity matrix $\mathbf{C}^t$:
\begin{equation}
\delta_{\bm{w}}(t) = \|\bm{w}^t - \bm{w}^{t-1}\|_2, \qquad
\delta_{\mathbf{C}}(t) = \|\mathbf{C}^t - \mathbf{C}^{t-1}\|_F.
\end{equation}
An utility score is defined as $u(t) = \lambda_w \delta_{\bm{w}}(t) + \lambda_C \delta_{\mathbf{C}}(t)$, where $\lambda$ balances the two factors.
Note that standard FL checkpointing ignores $\delta_{\mathbf{C}}$. However, in an interdependent NDT context, a sharp change in $\mathbf{C}^t$ (e.g., base station handovers) often triggers NDT re-clustering. Rolling back across a clustering boundary is mathematically complex and error-prone; therefore, we treat rounds where re-clustering occurs as structural anchors that must be checkpointed regardless of model or topology drift.

\myparatight{2) Elastic saving policy}
Let $\mathcal{M}$ denote the set of stored checkpoints on the global server, bounded by a budget $B$. The server persists a new full checkpoint $(t, \bm{w}^t)$ if:
(i) $u(t) \ge \tau_{\text{drift}}$ (indicating significant information gain);
(ii) A re-clustering event is triggered at round $t$; or
(iii) $t$ is a periodic keep-alive round, i.e. $t \bmod p_t = 0$.

To maximize storage efficiency during stable periods, the saving interval $p_t$ is updated online using an inverse-exponential schedule:
\begin{equation}
p_{t+1} = \text{clip}\left( p_t \cdot \exp\bigl(\kappa(\tau_{\text{drift}} - u(t))\bigr), ~ p_{\min}, ~ p_{\max} \right),
\end{equation}
where $\kappa$ controls responsiveness. This rule increases checkpoint sparsity when the twin model or topology is stable, and increases density when the system evolves rapidly.
If $|\mathcal{M}| > B$, we apply \textit{temporal coarsening}: the system retains dense checkpoints for the recent window (e.g., last $H$ rounds) to support fast handling of recent requests, while exponentially sparsifying older checkpoints (keeping one every $2^j$ rounds).

\myparatight{3) Proximal rollback retrieval}
When an untwinning request requires a rollback to round $T-K$, the exact checkpoint $\bm{w}^{T-K}$ may not exist in the sparse set $\mathcal{M}$.
Instead of storing high-dimensional update traces to reconstruct the state (which consumes as much storage as the model itself), we employ an elastic remapping strategy, where
the system retrieves the nearest \textit{preceding} checkpoint:
\begin{equation}
t^* = \max \{ t \in \mathcal{M} \mid t \le T-K \}, \qquad \hat{\bm{w}} = \bm{w}^{t^*}.
\end{equation}
The system then initiates the remapping process from $t^*$ rather than $T-K$.
This effectively extends the remapping phase from $K$ rounds to $K' = K + (T-K - t^*)$.
This approach trades a marginal increase in computational overhead (additional training rounds $T-K - t^*$) for significant storage reduction. Since $p_t$ is strictly bounded by $p_{\max}$, the worst-case additional computation is bounded, ensuring predictable system performance while reducing storage requirements from linear $\mathcal{O}(T)$ to logarithmic $\mathcal{O}(\log T)$.
\section{Theoretical Analysis on Untwinning Process}

This section provides theoretical analysis for our proposed \algO and \algM to guarantee the model integrity. We first present the following key assumptions about the model gradient estimators, which are common in related literature \cite{10234397,Fraboni2023sifu}. The proof is presented based on the NDT aggregation rule of Eq.~(\ref{eq:FedAvg}) as an example, and its arguments are easily generalized to other aggregation rules.
\begin{assumption}[$L$-Smoothness]
    A function $f : \mathbb{R}^d \to \mathbb{R}$ is considered $L$-smooth if it is differentiable and its gradient, $\nabla f : \mathbb{R}^d \to \mathbb{R}^d$, satisfies a Lipschitz continuity property with a constant $L$. Specifically, for any $\bm{w}_1, \bm{w}_2 \in \mathbb{R}^d$, the following inequality holds:
    \begin{equation}
    \|\nabla f(\bm{w}_1) - \nabla f(\bm{w}_2)\| \leq L \|\bm{w}_1 - \bm{w}_2\|.
    \end{equation}
\end{assumption}
\begin{assumption}[Bounded Stochastic Gradient Norm]
    The squared norm of the stochastic gradients is assumed to have a finite expected value. Formally, for each client $i = 1, \ldots, n$ and iteration $t = 1, \ldots, t$, there exists $A > 0$ such that:
    \begin{equation}
    \mathbb{E}_{\xi_t} \|\nabla f_i(\bm{w}_t)\|^2 \leq A^2,
    \end{equation}
    where $\xi_t$ represents the random sampling process. 
\end{assumption}
\begin{remark}
These assumptions are fundamental for establishing theoretical guarantees in optimization algorithms, particularly in distributed and stochastic settings. The $L$-smoothness assumption ensures that the gradient changes at a controlled rate, which is crucial for bounding the convergence rate and ensuring the stability of gradient-based methods. The bounded norm assumption limits the variability introduced by random sampling to ensure stable convergence. 
\end{remark}

While Assumptions 1 and 2 provide the necessary tractability for our theoretical framework, we acknowledge that real-world mobile networks may not strictly satisfy global $L$-smoothness, which is characterized by heterogeneous data and non-convex loss functions. 
However, these assumptions effectively approximate the \textit{local} behavior of the objective function near the convergence trajectory. 
Furthermore, the Bounded Gradient Norm can be practically enforced in our system through \textit{gradient clipping}, a standard engineering practice where the server or clients clip the norm of update vectors $\bm{g}(\bm{w}_n^t)$ to a threshold $C$ before aggregation.
This ensures that the theoretical bounds derived herein remain valid upper bounds for the worst-case divergence in the deployed NDTs.

\begin{theorem}
\label{thm:single_unlearning}
Let $\ell_i(\bm{w})$ be the local loss function of each NDT $i$, and its gradient is $L$-Lipschitz.  
Let $n_u$ be the target NDT, and let $[S_u]$ denote the set of NDTs that are strongly connected according to the connectivity matrix. 
For any iteration $t \leq T$, define the sensitivity degree as
\(
   \alpha(t, n_u) 
   \;=\;
   \bigl\|\,
   \bm{w}\bigl([N], \bm{w}^t\bigr)
   \;-\;
   \bm{w}\bigl(\,[N] \setminus [S_u],\,\bm{w}^t\bigr)
   \bigr\|,
\)
where $\bm{w}([N],\cdot)$ represents the global twin update from iteration~$t$ to $t+1$ when all NDTs $[N]$ participate, and $\bm{w}([N]\!\setminus\![S_u],\cdot)$ is the same update but excluding the set $[S_u]$. If there exists an upper bound function
\(
   \phi(t, n_u)\quad\text{such that}\quad
   \alpha(t, n_u)\;\le\;\phi(t, n_u),
\)
then adding Gaussian noise 
\(
  \xi \sim \mathcal{N}\bigl(\mathbf{0},\,\gamma^{2}\,I\bigr)
\) 
gives
\(
   \gamma 
   \;\ge\;
   \;\frac{\Omega \phi(t,n_u)}{\epsilon}
\)
to the final global twin model ensures $(\epsilon,\beta)$-indistinguishability for the target NDT $n_u$,
where $\Omega = \sqrt{2\,\bigl(\ln(1.25)-\ln(\beta)\bigr)}$.  
In other words, an external observer cannot distinguish whether $n_u$ had ever participated in the twinning process.
\end{theorem}

\begin{proof}

Let $\bm{w}^t$ be the global twin model up to iteration $t$ when all NDTs in $[N]$ are included.  Suppose we define a hypothetical model $\mathbf{u}^t$ when $[S_u]$ were excluded from every step up to iteration~$t$.  The corresponding model sensitivity
\(
   \alpha(t,n_u) \;=\; \|\bm{w}^t - \mathbf{u}^t\|
\)
captures how strongly $[S_u]$ can contribute to the twin model.
Under the $L$-smoothness assumption, each local gradient $\nabla \ell_i(\bm{w})$ is $L$-Lipschitz. Once the step size $\eta>0$ is chosen suitably, the difference between two twin models at iteration $t+1$ can be bounded by a factor times the difference at iteration~$t$. Formally, if we let $G(\bm{w},\mathcal{S})$ denote the aggregated gradient from subset $\mathcal{S}\subseteq[N]$ at $\bm{w}$, then from iteration~$t$ to $t+1$ we have
\(
   \bm{w}^{t+1} 
   \;=\;
   \bm{w}^t
   \;-\;
   \eta \, G(\bm{w}^t,[N]), 
\) and
\(
   \mathbf{u}^{t+1}
   \;=\;
   \mathbf{u}^t
   \;-\;
   \eta \, G(\mathbf{u}^t,\,[N]\!\setminus\![S_u]).
\) This gives us
\begin{equation}
\begin{split}
    \|\bm{w}^{t+1} - \mathbf{u}^{t+1}\|
   & \le \bigl(1 + \eta L\bigr)\,\|\bm{w}^t - \mathbf{u}^t\| \\
   & + \eta\,
   \Bigl\|\,
   G(\mathbf{u}^t,[N]) - G(\mathbf{u}^t,[N]\!\setminus\![S_u])
   \Bigr\|.
\end{split}
\end{equation}
Repeating this inequality over $t$ rounds yields a sum of the local influence from $[S_u]$.  We define
\(
   \phi(t,n_u)
   \;=\;
   \sum_{\tau=0}^{\,t-1}
   B^{\,t-1-\tau}\;\Delta_{[S_u]}(\tau),
\)
where $B=1+\eta L$ and
\(
   \Delta_{[S_u]}(\tau) 
   \;=\;
   \eta\,\bigl\|
   G\bigl(\mathbf{u}^\tau,[N]\bigr)
   -
   G\bigl(\mathbf{u}^\tau,[N]\!\setminus\![S_u]\bigr)
   \bigr\|.
\)
Thus, by induction, it follows that 
\(\alpha(t,n_u)\;=\;\|\bm{w}^t-\mathbf{u}^t\|\;\le\;\phi(t,n_u).\)  
Hence $\phi(t,n_u)$ stands as a valid bound on how far $\bm{w}^t$ can deviate from $\mathbf{u}^t$ due to $[S_u]$.
%
Recall that our \algO includes a step that rollback from the mapped twin model back to an earlier checkpoint $\bm{w}^{T-K}$, where $K$ is chosen so that the residual distance $\alpha(T-K, n_u)$ is below a threshold. Specifically, we define
\(
   K \;=\;\arg\max\nolimits_{k \le T}\,\Bigl[\alpha(T-k,n_u)\,\le\,\gamma^*\Bigr],
\)
or equivalently ensures $\phi(T-k,n_u)\le\gamma^*$.  In short, we find a checkpoint in the twinning process where $[S_u]$ either was absent or its influence was below a desired threshold.  Once we revert to $\bm{w}^{T-K}$ (and optionally add some small Gaussian noise), $[S_u]$ is effectively untwinned from the current global NDT, because any potential contributions from $n_u$ or its strongly connected neighbors do not propagate beyond iteration $T-K$ in the new timeline.
%
To solidify untwinning process, we utilize the standard Gaussian mechanism~\cite{dwork2014algorithmic}.  Let $\Delta$ be the $\ell_2$-sensitivity of the twin model and $\Delta = \|\bm{w}^t - \mathbf{u}^t\|\le\phi(t,n_u)$.  By adding Gaussian noise 
\(
   \mathcal{N}(\mathbf{0},\,\gamma^2 I), 
\)
the untwinned model obeys $(\epsilon,\beta)$-indistinguishability.  In simpler terms, an observer cannot decide whether $[S_u]$ was removed or not, because the distribution of the untwinned twin models differs from map-from-scratch by at most the factor in the $(\epsilon,\beta)$-indistinguishability.
\end{proof}


We now extend Theorem~\ref{thm:single_unlearning} to the case where \emph{multiple clusters} each have a set of NDTs to remove, as shown in Algorithm~\ref{federated_untwinning_multiple}.  
To avoid redundancy, we refer to the bounding arguments and rollback rationale from the proof of Theorem~\ref{thm:single_unlearning}.  
Here, we highlight the extra details needed to accommodate parallel untwinning across clusters.

\vspace*{1ex}
\begin{theorem}
\label{thm:parallel_unlearning}
Suppose the entire NDT system is split into $A$ clusters, indexed by $a=1,\dots,A$.  
Cluster $a$ untwins a set $[U_a]$ of NDTs.  
Under the same assumptions used in Theorem~\ref{thm:single_unlearning}, and with a Gaussian noise variance of
\(
  \gamma 
  \;\;\ge\;
  \frac{\Omega \phi_a(\cdot)}{\epsilon},
\)
the global twin model $\bm{w}^{T'}$ is $(\epsilon,\beta)$-indistinguishable for
\(
  \bigcup_{a=1}^A [U_a].
\)
\end{theorem}

\begin{proof}

As in the proof of Theorem~\ref{thm:single_unlearning}, we define that for each cluster $a$, a function $\phi_a(t)$ bounding the $\ell_2$‐distance between a global twin model mapped \emph{with} $[U_a]$ vs.\ \emph{without} $[U_a]$.  
The same Lipschitz and smoothness arguments guarantee 
\(
  \|\bm{w}^t - \mathbf{v}_a^t\|
  \;\le\;
  \phi_a(t),
\)
where $\mathbf{v}_a^t$ is the hypothetical model had $[U_a]$ never contributed.
\algM rolls back the global twin model to $\bm{w}^{T-K_a}$, injects noise proportional to $\phi_a(T-K_a)$,
and remaps from there, excluding the twin model and data from $[U_a]$.
Each cluster $a$ does this process in parallel.  
Once cluster $a$ finishes rollback, it untwins $[U_a]$ in the sense of $(\epsilon,\beta)$‐indistinguishability, because any prior traces of $[U_a]$ are now obfuscated by noise.
We rely on an inductive approach:
\begin{list}{\labelitemi}{\leftmargin=1em \itemindent=-0.08em \itemsep=.2em}
\item \emph{Base case (first cluster).}  After the first cluster completes rollback and noise-adding step to untwin $[U_1]$, the resulting model is $(\epsilon,\beta)$-indistinguishable from the map-from-scratch twin model.
\item \emph{Inductive step (subsequent clusters).}  Suppose that after clusters $1,\dots,a-1$ have untwinned $[U_1],\dots,[U_{a-1}]$, the model is $(\epsilon,\beta)$-indistinguishable from map-from-scratch twin model $\bigcup_{j=1}^{a-1}[U_j]$.  
Cluster $a$ then does the same step to remove $[U_a]$.  
Because $[U_1],\dots,[U_{a-1}]$ were already excluded, no future mapping round reintroduces them.  
Therefore, the updated twin model is $(\epsilon,\beta)$-indistinguishable.
\end{list}
Iterating to $a=A$, we conclude that after all clusters have removed their target NDTs, the final NDT model is $(\epsilon,\beta)$-indistinguishable from map-from-scratch twin model.
\end{proof}

\section{Experiment and Validation}
\label{sec:exp}

\subsection{Experimental Setup}

\subsubsection{Scenario Configurations}

The proposed approach assumes cooperation, central control, and willingness to share rollback states, but the paper does not convincingly argue how this would be adopted in realistic competitive network environments.

We evaluate the performance of our proposed untwinning mechanisms using a real-world traffic data scenario from a segment of the I-15 freeway in Utah, U.S. 
Our collected dataset\footnote{The dataset has been uploaded to the Department of Transportation and is available upon request~\cite{udot2025traffic}.} aggregates traffic data from inductive loop detectors installed along the freeway, as shown in Fig.~\ref{fig:data}. 
These deployed sensors include both single-loop and dual-loop detectors, which continuously monitor traffic conditions by measuring vehicle counts and speeds at fixed locations along the roadway.
Several local NDTs are constructed to monitor vehicle traffic and speed in real-time, covering different highway areas. A few cluster-level NDTs are aggregated from local NDTs to expedite the parallel request untwinning in \algM. One global NDT is deployed on a roadside server to coordinate the entire twinning and untwinning process.

\begin{figure*}[!h]
    \centering
    \includegraphics[width=0.85\textwidth]{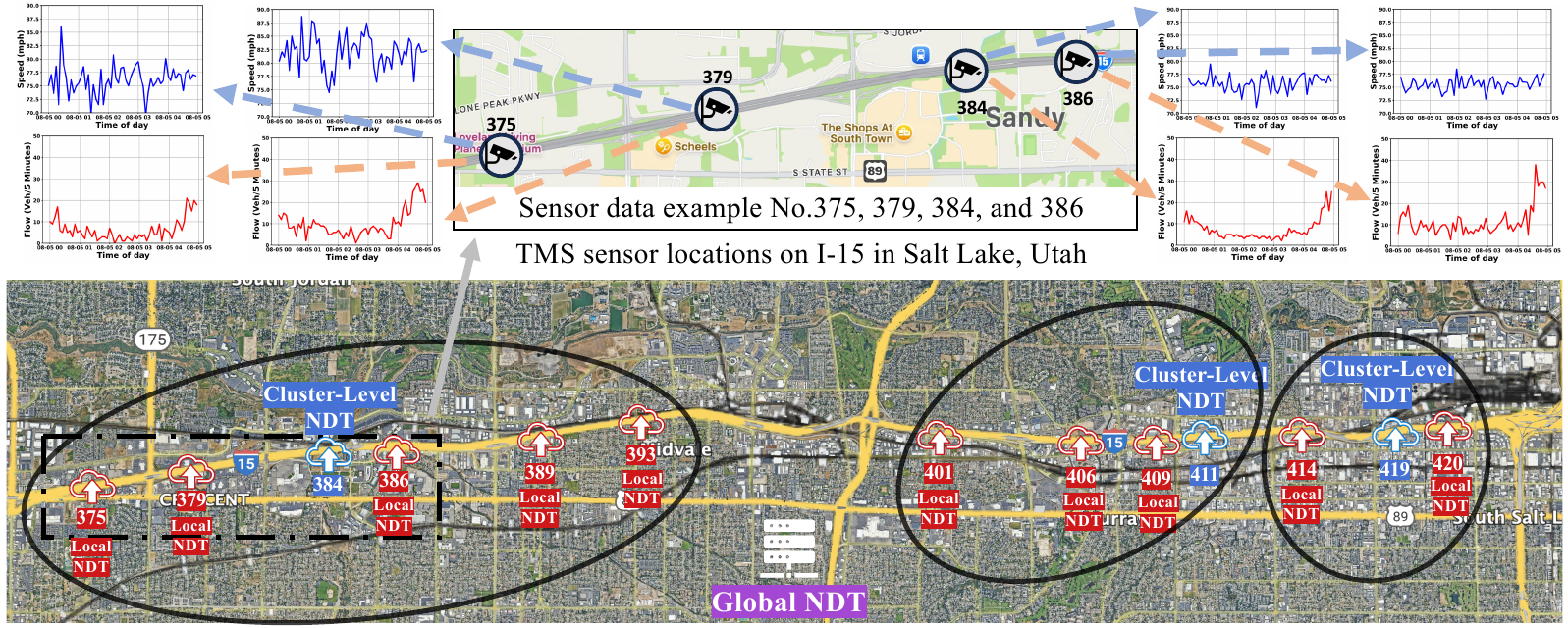}
    \caption{Testbed configurations on I-15 in Salt Lake, Utah.}
    \vspace{-.1in}
    \label{fig:data}
\end{figure*}

\subsubsection{Untwinning Benchmarks}

Since we are the first work to explore the untwinning process in sensor networks, there are no baselines that can be directly applied to our settings. 
Despite detailed operational differences, federated machine unlearning serves as a valid baseline due to its strong high-level conceptual alignment with our framework. Specifically, both paradigms operate within analogous distributed learning scenarios, share fundamental assumptions regarding data distributions, and aim to eliminate targeted data influences from an aggregated global model.
We adapt and re-implement six baseline unlearning methods, including SIFU and IFU~\cite{Fraboni2023sifu}, Crab~\cite{jiang2024towards}, $\text{FedME}^2$~\cite{10234397}, FedEraser~\cite{liu2020federaser}, FedRecovery~\cite{10189868}, and  FedAccum~\cite{liu2020federaser}, to our untwinning framework, and then illustrate the effectiveness of our process-centric methods. When receiving multiple parallel requests, these baseline methods process each untwinning request sequentially.

\subsubsection{Evaluation Metrics}

Our evaluation framework employs two principal metrics. The first one is prediction error measured by mean squared error (MSE), which rigorously quantifies the \textbf{twin model accuracy} by averaging the squared differences between the predicted results from the NDT and the observed traffic flow and vehicle information collected from 21 inductive loop detectors. 
The second one is the \textbf{twinning efficiency}, which measures the overall computational time required to complete the untwinning process.

\subsection{Numerical Results}

\myparatight{\textit{1) Untwinning strongly-connected NDTs is necessary}}
Figure~\ref{fig:connected} illustrates the necessity of untwinning both target NDT and its propagating influence, i.e., those strongly connected neighbors. Each marker represents a different map-from-scratch strategy: \textit{None} (no untwinning performed), \textit{Target Only} (untwinning only target NDT), and \textit{Connected Set} (untwinning target NDT plus its strongly connected neighbors). Red markers indicate vehicle speed data, while blue markers correspond to traffic flow data. 
X-axis shows the prediction error (i.e., MSE) from global twin model evaluated on the covered data from the remaining NDTs, whereas y-axis measures the prediction error testing on the covered data from target NDT. For both vehicle speed and traffic flow forecasting, untwinning only target NDT results in a slight increase in prediction error, indicating that correlated contextual information associated with target NDT persists in global twin through its strongly connected counterparts. 
Untwinning only target NDT leaves residual influence in global twin, as it reflects dynamics from correlated sensing areas. In contrast, untwinning multiple interconnected NDTs more effectively removes this influence, significantly increasing prediction error, for example, traffic flow error rises from 0.382 to 0.541, thus ensuring complete data removal.

\begin{figure}[!h]
    \centering
    \includegraphics[width=0.43\textwidth]{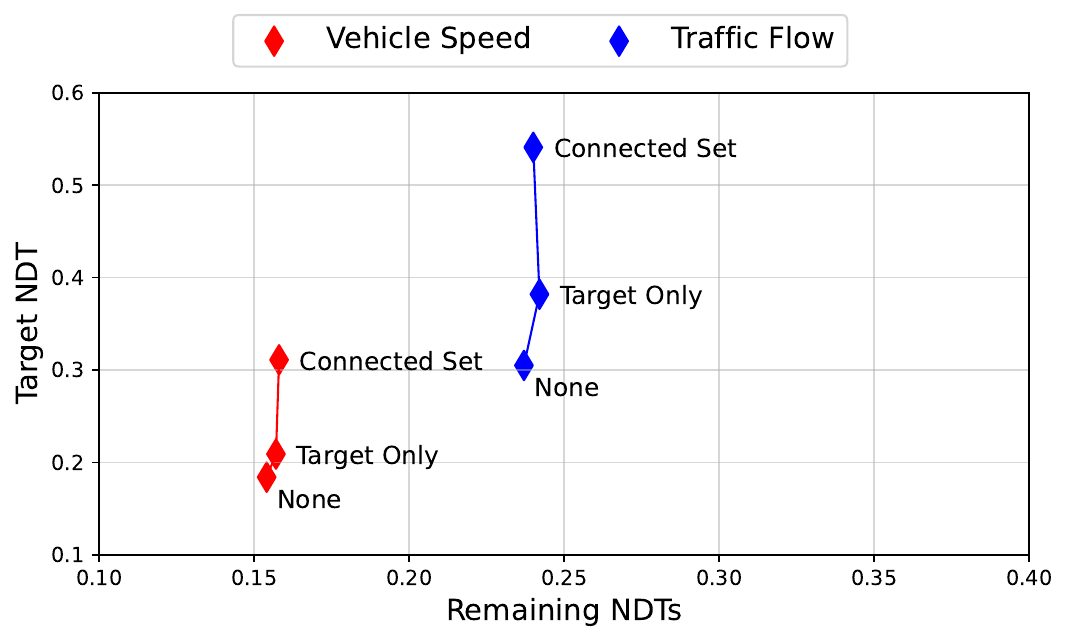}
    \caption{Comparison between untwinning target NDT only and untwinning strongly connected NDT set upon request.}
    \label{fig:connected}
\end{figure}

\myparatight{\textit{2) \algO is effective in small-scale settings}}
Fig.~\ref{fig:single_traffic} and \ref{fig:single_speed} evaluate the prediction error differences (PED) of untwinned models from our methods and the map-from-scratch twin (ground-truth model).
They demonstrate that \algO consistently achieves the lowest differences across both vehicle speed and traffic flow predictions when receiving only one untwinning request. 
Significant lower PED values are observed in the evaluations on target NDTs, shown in blue bars, while moderate improvements are illustrated in the testing on remaining NDTs, shown in orange bars, indicating that our proposed \algO more effectively isolates the impact of the target twins.
This is primarily due to its precise untwinning mechanism, which effectively removes the contributions of a target NDT and its propagating influence. 
By computing connectivity scores, \algO identifies not only the direct influence of target NDT but also the indirect effects from its correlated neighbors. 
\algO then determines an optimal rollback checkpoint by evaluating the contribution from each target NDT, ensuring that global twin model is rolled back to the state where the unwanted influences are most prominent. 
Furthermore, the injection of Gaussian noise during the rollback process further obfuscates any residual information, enabling a clean remapping phase that preserves beneficial contributions from those weakly connected NDTs that cover other adjacent sub-networks. 
This context-aware approach yields a robust model with minimal PED, confirming that \algO outperforms baseline methods in response to a single untwinning request while remaining indistinguishable from map-from-scratch twin models—consistent with our theoretical guarantee in Sec.~4.

\begin{figure}[t]
    \centering
    \footnotesize
    \begin{subfigure}[b]{0.23\textwidth}
        \centering
        \includegraphics[width=\textwidth]{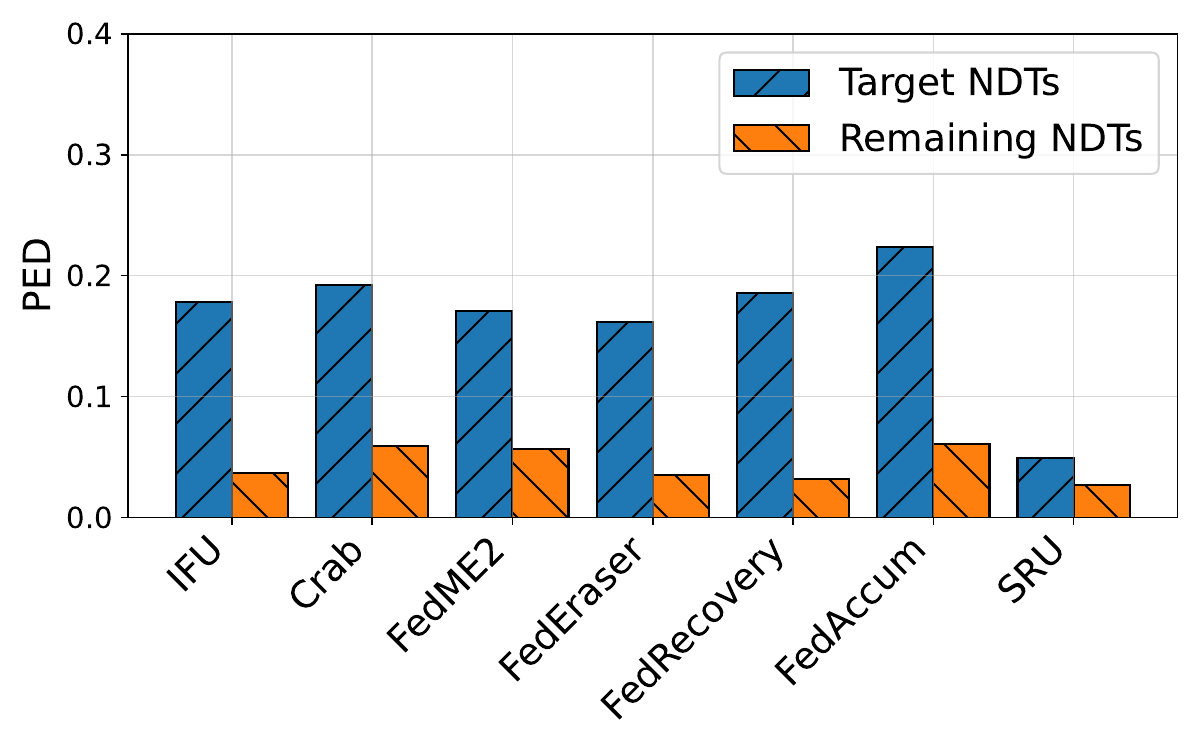}
        \caption{Traffic flow prediction error with one request.}
        \label{fig:single_traffic}
    \end{subfigure}
    \hfill
    \begin{subfigure}[b]{0.23\textwidth}
        \centering
        \includegraphics[width=\textwidth]{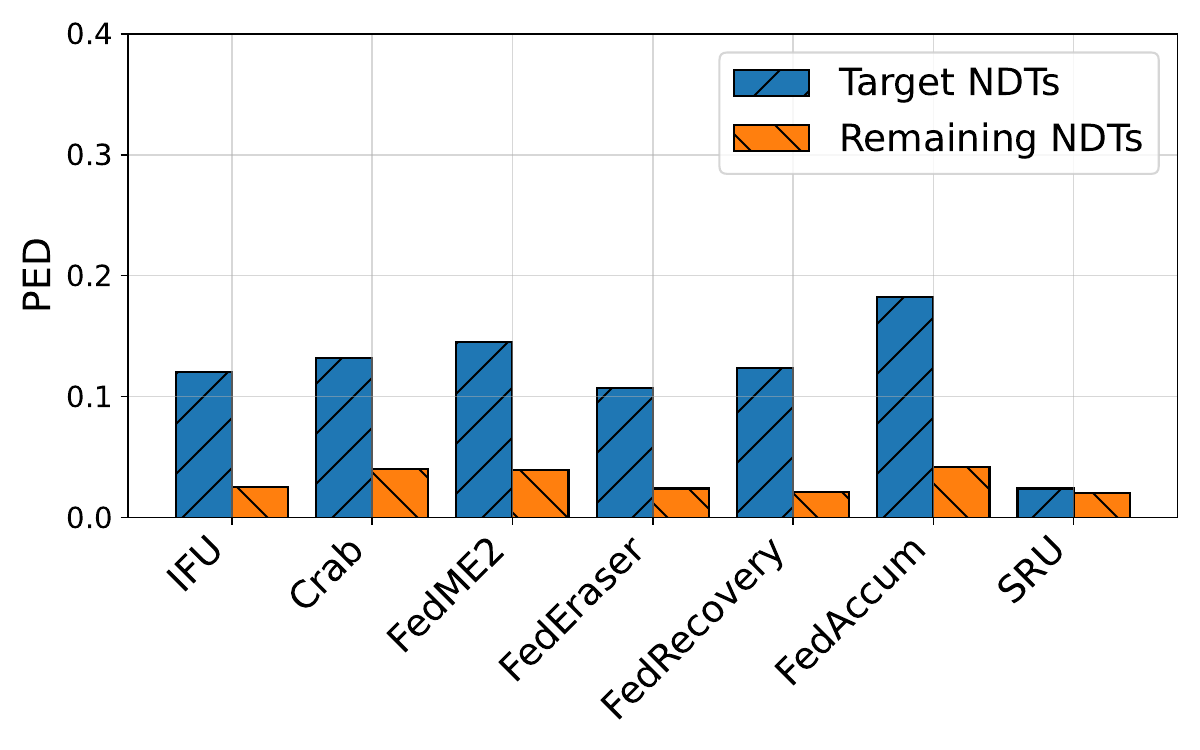}
        \caption{Vehicle speed prediction error with one request.}
        \label{fig:single_speed}
    \end{subfigure}
    
    \begin{subfigure}[b]{0.23\textwidth}
        \centering
        \includegraphics[width=\textwidth]{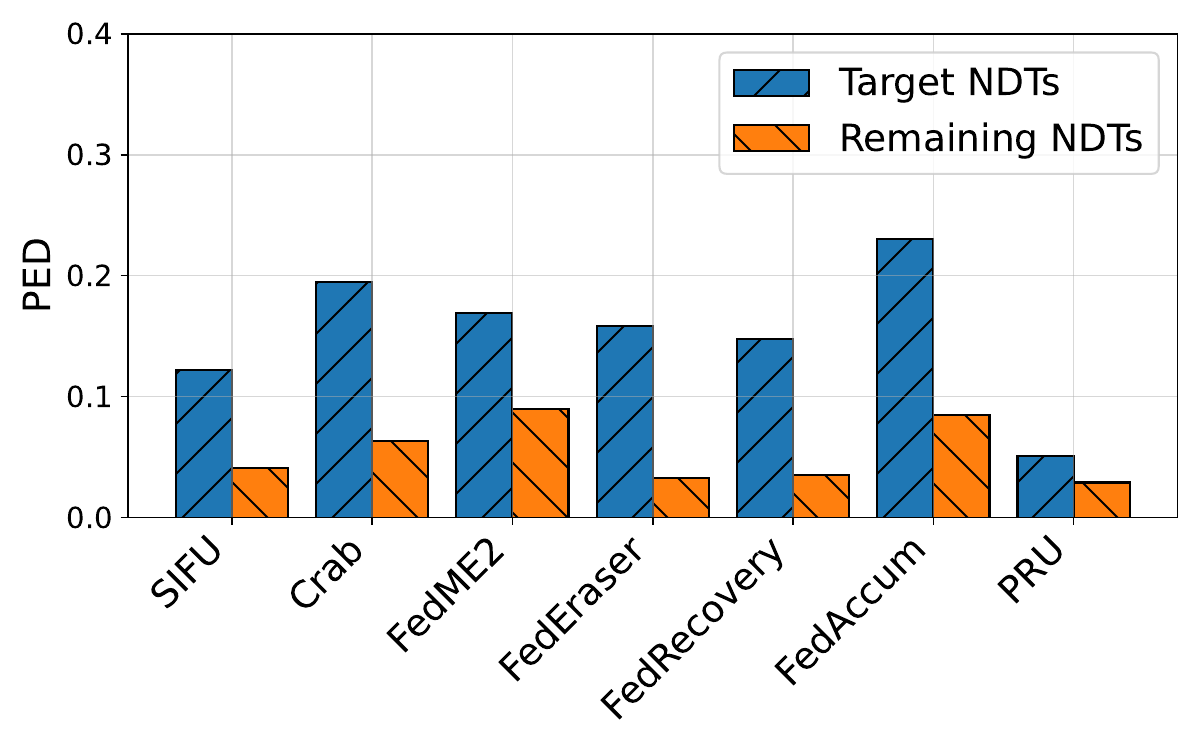}
        \caption{Traffic flow prediction error with three requests.}
        \label{fig:parallel_traffic}
    \end{subfigure}
    \hfill
    \begin{subfigure}[b]{0.23\textwidth}
        \centering
        \includegraphics[width=\textwidth]{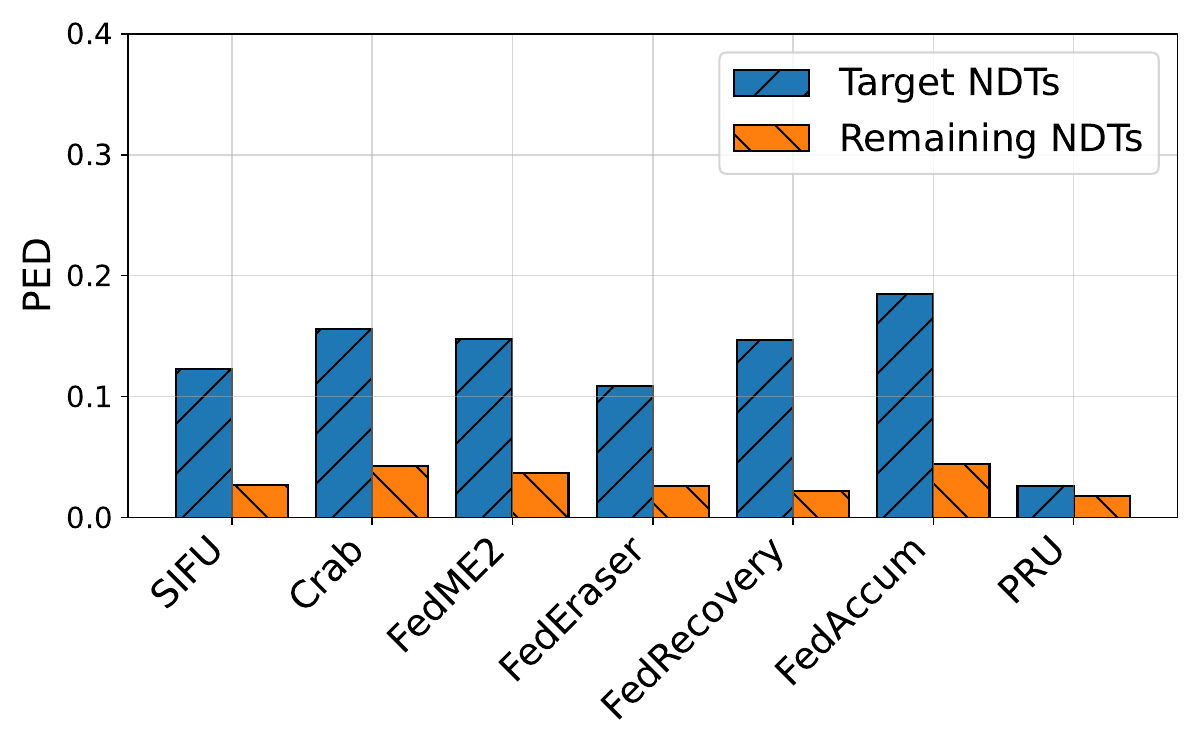}
        \caption{Vehicle speed prediction error with three requests.}
        \label{fig:parallel_speed}
    \end{subfigure}

    \caption{PED under diverse request settings.}
    \label{fig:combined_prediction}
\end{figure}

\myparatight{\textit{3) \algM efficiently handles multiple parallel requests}}
Fig.~\ref{fig:parallel_traffic} and Fig.~\ref{fig:parallel_speed} demonstrate that \algM consistently achieves lower PED from map-from-scratch twin models than the baseline methods under the scenario of multiple simultaneous untwinning requests.  
As expected in both scenarios, with three requests processed concurrently by default, the blue bars for target NDTs exhibit notably reduced prediction errors, while the orange bars for remaining NDTs also show appreciable improvements with \algM. 
These performance gains can be attributed to its effective network clustering of NDTs and the coordinated rollback mechanism, which carefully determines a \textit{per-cluster} rollback checkpoint to excise the contributions of all target NDTs and their propagating influences. 
Especially in cases where multiple untwinning requests occur within the same cluster, the system processes them collectively, thereby reducing overall time consumption.

\begin{table}[!ht]
\centering
\scriptsize
\caption{Results on untwinning overhead and scalability.}
\label{tab:combined}

\begin{subtable}[t]{0.22\textwidth}
\centering
\caption{Impact of NDT density.}
\begin{tabular}{@{}c|c|c|c@{}}
\toprule
\textbf{} & \textbf{10} & \textbf{20} & \textbf{30} \\
\midrule
Target & 0.051 & 0.043 & 0.045 \\
Remain & 0.029 & 0.020 & 0.023 \\
Time   & 251   & 480   & 698 \\
\bottomrule
\end{tabular}
\end{subtable}
\begin{subtable}[t]{0.22\textwidth}
\centering
\caption{Request ptg.}
\begin{tabular}{@{}c|c|c|c@{}}
\toprule
\textbf{} & \textbf{20\%} & \textbf{30\%} & \textbf{40\%} \\
\midrule
Target & 0.043 & 0.051 & 0.055 \\
Remain & 0.021 & 0.029 & 0.031 \\
Time   & 175   & 251   & 322 \\
\bottomrule
\end{tabular}
\end{subtable}
\end{table}

\noindent \textbf{\textit{4) PRU scales well with the NDT and request density:}} Table~\ref{tab:combined}a shows that the number of deployed NDTs does not significantly affect the untwinning performance of \algMns. In our default configuration, the sensor network comprises 10 NDTs and processes 3 concurrent untwinning requests. As the network scales, we evaluate a scenario in which 30\% of the NDTs are subject to untwinning. Despite the increased number of requests, the observed difference from map-from-scratch twins remains consistent across all configurations. Notably, the untwinning time does not increase linearly with network size—likely due to multiple requests occurring within the same cluster, enabling the clustering and scheduling mechanisms to optimize processing. This integrated approach produces a global twin model that remains indistinguishable from a fully remapped version while minimizing untwinning time.
Furthermore, Table~\ref{tab:combined}b shows that the number of untwinning requests has minimal impact on twin model accuracy. While the request count increases, the processing time is below the expected linear increase, since multiple requests within the same cluster are handled together.

\myparatight{\textit{5) Trade-offs between Accuracy and Computation}}
Table~\ref{tab:cost_tradeoff} shows the critical trade-offs between maintaining twin accuracy and reducing system costs. Standard remapping provides the upper-bound accuracy, but it is not practical due to its long runtime of over 3 hours and high storage demand of 850.5 MB. Baseline methods attempt to lower this latency but often sacrifice too much performance. For instance, SIFU achieves a fast runtime of 199 seconds but suffers from a massive drop in accuracy marked by a +323.81\% gap. Additionally, most baselines still require over 780 MB of storage to maintain historical states, which is nearly as expensive as full retraining. Our proposed methods balance these factors effectively. By using precise rollback with adaptive checkpointing, PRU reduces runtime to 132 seconds, which is roughly 94 times faster than retraining. It also drastically lowers storage costs to 42.5 MB. This represents only about 5\% of the storage required by other methods while maintaining high fidelity of the NDTs with a marginal accuracy gap of +0.48\%.

\begin{table}[t]
\centering
\caption{Comparison of computational overhead, parameter storage cost, and model fidelity with baselines.}
\label{tab:cost_tradeoff}
\resizebox{\columnwidth}{!}{%
\begin{tabular}{lcccc}
\toprule
\textbf{Method} & \textbf{\begin{tabular}[c]{@{}c@{}}Avg. Runtime\\ (sec)\end{tabular}} & \textbf{\begin{tabular}[c]{@{}c@{}}Storage Cost\\ (MB)\end{tabular}} & \textbf{\begin{tabular}[c]{@{}c@{}}MSE \end{tabular}} & \textbf{\begin{tabular}[c]{@{}c@{}}Acc. Gap\\ (vs. Remap)\end{tabular}} \\ \midrule
Remapping & 12,450 & 850.5 & 0.0210 & 0.00\% \\
FedEraser & 981 & 842.1 & 0.0315 & +50.00\% \\
Crab & 848 & 835.4 & 0.0342 & +62.86\% \\
$\mathbf{FedME^2}$ & 760 & 815.8 & 0.0410 & +95.24\% \\
FedRecovery & 629 & 805.2 & 0.0525 & +150.00\% \\
FedAccum & 491 & 798.5 & 0.0683 & +223.81\% \\
SIFU & 199 & 780.1 & 0.0891 & +323.81\% \\ \midrule
\textbf{SRU (Ours)} & 172 & \textbf{41.2} & 0.0212 & +0.95\% \\
\textbf{PRU (Ours)} & \textbf{132} & 42.5 & \textbf{0.0211} & \textbf{+0.48\%} \\ \bottomrule
\end{tabular}%
}
\end{table}

\myparatight{\textit{6) Adaptive Checkpointing Enhances the Unitwinning Efficiency}}
Table~\ref{tab:adaptive_checkpoint} shows an ablation study that highlights the efficiency of the proposed ATAC module (Sec. III-D). 
A naive strategy saves the global twin model at every round to ensure exact recovery, but this consumes a massive 42.50 GB of parameter storage. 
While increasing the fixed saving interval $p$ to 50 rounds reduces storage to 0.85 GB, it incurs a high \textit{replay overhead}--the averaged additional computation time required to reconstruct a specific state from the nearest preceding checkpoint. 
Because a fixed interval ignores network structural changes, the gap between the target round and the last saved checkpoint can be large, driving the replay overhead up to 5.8 seconds. 
In contrast, our ATAC dynamically adjusts the saving frequency based on model stability and network topology shifts. It retains only 48 high-utility checkpoints, achieving a 95.2\% storage reduction compared to the naive approach. Crucially, because ATAC intelligently preserves snapshots during rapid network changes, it keeps the replay overhead to a manageable 2.1 seconds, successfully balancing storage efficiency with fast recovery.

\begin{table}[t]
\centering
\caption{Efficiency enhancement by ATAC module.}
\label{tab:adaptive_checkpoint}
\resizebox{\columnwidth}{!}{%
\begin{tabular}{lcccc}
\toprule
\textbf{Checkpointing Strategy} & \textbf{\begin{tabular}[c]{@{}c@{}}Stored\\ Ckpts\end{tabular}} & \textbf{\begin{tabular}[c]{@{}c@{}}Storage\\ (GB)\end{tabular}} & \textbf{\begin{tabular}[c]{@{}c@{}}Storage\\ Reduct.\end{tabular}} & \textbf{\begin{tabular}[c]{@{}c@{}}Replay\\ Overhead (s)\end{tabular}} \\ \midrule
Naive (Every Round) & 1,000 & 42.50 & - & 0.0 \\
Fixed Interval ($p=10$) & 100 & 4.25 & 90.0\% & 1.2 \\
Fixed Interval ($p=50$) & 20 & 0.85 & 98.0\% & 5.8 \\ \midrule
\textbf{ATAC (Ours)} & \textbf{48} & \textbf{2.04} & \textbf{95.2\%} & \textbf{2.1} \\ \bottomrule
\end{tabular}%
}
\end{table}

\section{Related Works}
\label{sec:related}

\noindent \textbf{Network Digital Twins.} Several works highlight how NDTs are central to next-generation network systems, particularly for optimizing 5G deployments to boost both coverage and throughput \cite{networkdigitaltwin, digitaltwin5g, zhang2025synergizing,tang2025research,tao2024ran,li2024intelligent,tan2025energy, graphneuralnetwork}. \cite{connectingtwins} underscores vital need for solid data collection and real-time analytics in designing effective NDTs. 
Within vehicular networks, \cite{intelligentdigitaltwin, Ding2024digital} demonstrates that NDTs enhance reliability through advanced coordination strategies.  
Most of these studies, however, concentrate on how NDTs optimize network systems but do not delve deeply into detailed construction of NDTs for physical assets and properties~\cite{khan2022digital, Yang2024ICASSP}. Some recent research does explore underlying processes of constructing NDTs. For example, \cite{villa2024colosseum} proposes Colosseum, a testbed that fuses real-world experiments with emulated environments to refine wireless systems. 
In contrast, our approach introduces a cost-efficient strategy for reconstructing DTs in a backward direction, enabling targeted removal of undesired or deprecated NDT contributions. 

\vspace{+0.1cm}
\noindent \textbf{Machine Unlearning.} Machine unlearning has become an essential approach for enhancing privacy and security in machine learning, particularly within wireless and IoT networks~\cite{cao2015towards, bourtoule2020machineunlearning}. 
By allowing specific data points to be effectively removed from trained models, machine unlearning is vital for privacy protection, handling user disconnections, and mitigating data poisoning threats. 
FL has extended machine learning to decentralized settings, enabling devices to train models collaboratively without sharing their data~\cite{mcmahan2017communication}. 
However, as clients may wish to erase their data from global model, FL encounters unique challenges, since data contributions are distributed across various devices. Federated unlearning addresses these challenges by providing mechanisms to remove a client's impact from model~\cite{halimi2023federatedunlearningefficientlyerase, liu2020federaser, Fraboni2023sifu}.
NDTs have been used to bolster federated unlearning in vehicular and V2X communication networks, promoting low-latency, secure data interactions in these environments~\cite{yuan2024iot, Daluwatta2024dtfu}. 
Inspired by these existing studies, our work takes a further step by casting untwinning as a native capability of NDT systems, rather than merely a complementary module for supporting federated unlearning. Our method removes network entities together with their scenario-level influence under real-time synchronization, thereby preserving model fidelity while incurring minimal untwinning cost.

\section{Conclusion}

In this paper, we propose a network digital untwinning framework that selectively eliminates target NDTs with their propagating influences while maintaining model integrity. 
Our context-aware approach comprises both \algO and \algM fashions, where \algO removes target NDTs and their strongly connected neighbors through a joint optimal rollback and noise injection mechanism, while 
\algM extends this to handle multiple requests simultaneously via network-level clustering and untwining scheduling. 
Extensive experiments on real-world traffic data confirm that our proposed framework produces results indistinguishable from the map-from-scratch twin models while significantly improving operational efficiency.

\section*{Acknowledgment}
This research was supported by NSF through Award CNS--2440756, CNS--2312138, and SaTC-2350075.
\clearpage

\balance

\bibliographystyle{IEEEtran}
\bibliography{ref}

@inproceedings{cao2015towards,
    author={Cao, Yinzhi and Yang, Junfeng},
    booktitle={IEEE Security \& Privacy}, 
    title={Towards Making Systems Forget with Machine Unlearning}, 
    year={2015}
}

@inproceedings{bourtoule2020machineunlearning,
    title={Machine Unlearning},
    author={Lucas Bourtoule and Varun Chandrasekaran and Christopher A. Choquette-Choo and Hengrui Jia and Adelin Travers and Baiwu Zhang and David Lie and Nicolas Papernot},
    booktitle={IEEE Security \& Privacy},
    year={2021}
}

@InProceedings{mcmahan2017communication,
  title = 	 {{Communication-Efficient Learning of Deep Networks from Decentralized Data}},
  author = 	 {McMahan, Brendan and Moore, Eider and Ramage, Daniel and Hampson, Seth and Arcas, Blaise Aguera y},
  booktitle = 	 {AISTATS},
  year = 	 {2017},
}

@inproceedings{halimi2023federatedunlearningefficientlyerase,
      title={Federated Unlearning: How to Efficiently Erase a Client in FL?}, 
      author={Anisa Halimi and Swanand Kadhe and Ambrish Rawat and Nathalie Baracaldo},
      year={2022},
      booktitle={ICML}
}

@inproceedings{liu2020federaser,
    author={Liu, Gaoyang and Ma, Xiaoqiang and Yang, Yang and Wang, Chen and Liu, Jiangchuan},
    booktitle={IEEE IWQoS}, 
    title={FedEraser: Enabling Efficient Client-Level Data Removal from Federated Learning Models}, 
    year={2021},
}

@inproceedings{Fraboni2023sifu,
  title={Sequential Informed Federated Unlearning: Efficient and Provable Client Unlearning in Federated Optimization},
  author={Fraboni, Yann and Vidal, Richard and Kameni, Laetitia and Lorenzi, Marco},
  booktitle={AISTATS},
  year={2024}
}

@inproceedings{yuan2024iot,
  author={Yuan, Yanli and Wang, Bingbing and Zhang, Chuan and Xiong, Zehui and Li, Chunhai and Zhu, Liehuang},
  booktitle={IEEE IoT-J}, 
  title={Toward Efficient and Robust Federated Unlearning in IoT Networks}, 
  year={2024},
}

@inproceedings{Daluwatta2024dtfu,
  author={Daluwatta, Wathsara and Edirimannage, Shehan and Elvitigala, Charitha and Khalil, Ibrahim and Atiquzzaman, Mohammed},
  booktitle={IEEE CommMag}, 
  title={DT-FU: Digital Twin-Driven Federated Unlearning for Resilient Vehicular Networks in the 6G Era}, 
  year={2024}
}

@inproceedings{10189868,
  author={Zhang, Lefeng and Zhu, Tianqing and Zhang, Haibin and Xiong, Ping and Zhou, Wanlei},
  booktitle={IEEE TIFS}, 
  title={FedRecovery: Differentially Private Machine Unlearning for Federated Learning Frameworks}, 
  year={2023}
}

@inproceedings{jiang2024towards,
  title={Towards Efficient and Certified Recovery from Poisoning Attacks in Federated Learning},
  author={Jiang, Yu and Shen, Jiyuan and Liu, Ziyao and Tan, Chee Wei and Lam, Kwok-Yan},
  booktitle={arXiv preprint arXiv:2401.08216},
  year={2024}
}

@inproceedings{10234397,
    author={Xia, Hui and Xu, Shuo and Pei, Jiaming and Zhang, Rui and Yu, Zhi and Zou, Weitao and Wang, Lukun and Liu, Chao},
    booktitle={IEEE JSAC}, 
    title={FedME2: Memory Evaluation \& Erase Promoting Federated Unlearning in DTMN}, 
    year={2023},
}

@techreport{Y3090,
    type =      {Y. 3090 Recommendations},
    institution =   {International Telecommunication Union},
    publisher = {International Telecommunication Union},
    title =     {Digital twin network – Requirements and architecture},
    year = 2024
}

@techreport{Y3091,
    type =      {Y. 3091 Recommendations},
    institution =   {International Telecomm. Union},
    publisher = {International Telecomm. Union},
    title =     {Digital twin network - Capability levels and evaluation methods},
    year = 2024
}

@inproceedings{wu2021digital,
  title={Digital twin networks: A survey},
  author={Wu, Yiwen and Zhang, Ke and Zhang, Yan},
  booktitle={IEEE IoT-J},
  year={2021}
}

@article{zhang2025digitalmag,
  title={Digital network twins for next-generation wireless: Creation, optimization, and challenges},
  author={Zhang, Zifan and Peng, Zhiyuan and Yu, Hanzhi and Chen, Mingzhe and Liu, Yuchen},
  journal={IEEE network},
  year={2025},
  publisher={IEEE}
}

@inproceedings{elayan2021digital,
  title={Digital twin for intelligent context-aware IoT healthcare systems},
  author={Elayan, Haya and Aloqaily, Moayad and Guizani, Mohsen},
  booktitle={IEEE IoT-J},
  year={2021}
}

@inproceedings{hakiri2024comprehensive,
  title={A comprehensive survey on digital twin for future networks and emerging Internet of Things industry},
  author={Hakiri, Akram and Gokhale, Aniruddha and Yahia, Sadok Ben and Mellouli, Nedra},
  booktitle={Computer Networks},
  year={2024}
}

@inproceedings{wang2024smart,
  title={Smart mobility digital twin based automated vehicle navigation system: A proof of concept},
  author={Wang, Kui and Li, Zongdian and Nonomura, Kazuma and Yu, Tao and Sakaguchi, Kei and Hashash, Omar and Saad, Walid},
  booktitle={IEEE TIV},
  year={2024}
}

@inproceedings{khan2022digital,
  title={Digital twin of wireless systems: Overview, taxonomy, challenges, and opportunities},
  author={Khan, Latif U and Han, Zhu and Saad, Walid and Hossain, Ekram and Guizani, Mohsen and Hong, Choong Seon},
  booktitle={IEEE ComST},
  year={2022}
}

@inproceedings{tao2024wireless,
  title={Wireless network digital twin for 6g: Generative ai as a key enabler},
  author={Tao, Zhenyu and Xu, Wei and Huang, Yongming and Wang, Xiaoyun and You, Xiaohu},
  booktitle={IEEE WCM},
  year={2024}
}

@inproceedings{zhang2024digital,
  title={Digital Twin-Assisted Data-Driven Optimization for Reliable Edge Caching in Wireless Networks},
  author={Zhang, Zifan and Liu, Yuchen and Peng, Zhiyuan and Chen, Mingzhe and Xu, Dongkuan and Cui, Shuguang},
  booktitle={IEEE JSAC},
  year={2024}
}

@inproceedings{villa2024colosseum,
  title={Colosseum as a digital twin: Bridging real-world experimentation and wireless network emulation},
  author={Villa, Davide and Tehrani-Moayyed, Miead and Robinson, Clifton Paul and Bonati, Leonardo and Johari, Pedram and Polese, Michele and Melodia, Tommaso},
  booktitle={IEEE TMC},
  year={2024}
}

@inproceedings{zhang2024mapping,
    author = {Zifan Zhang and Mingzhe Chen and Zhaohui Yang and Yuchen Liu},
    title = {Mapping Wireless Networks into Digital Reality through Joint Vertical and Horizontal Learning},
    booktitle = {IFIP/IEEE Networking},
    year = {2024}
}

@article{zhang2024securing,
  title={Securing distributed network digital twin systems against model poisoning attacks},
  author={Zhang, Zifan and Fang, Minghong and Chen, Mingzhe and Li, Gaolei and Lin, Xi and Liu, Yuchen},
  journal={IEEE IoT-J},
  year={2024},
  publisher={IEEE}
}

@inproceedings{Yang2024ICASSP,
  title={Optimizing Synchronization Delay for Digital Twin over Wireless Networks},
  author={{Yang, Zhaohui and Chen, Mingzhe and Liu, Yuchen and Zhang, Zhaoyang}},
  booktitle={IEEE ICASSP},
  year={2024}
}

@inproceedings{Ding2024digital,
  title={Joint Vehicle Connection and Beamforming Optimization in Digital Twin Assisted Integrated Sensing and Communication Vehicular Networks},
  author={Ding, W. and Zhang and Chen, M. and Liu, Y. and Shikh-Bahaei, M.},
  booktitle={IEEE IoT-J},
  year={2024}
}

@inproceedings{networkdigitaltwin,
	title        = {Network digital twin: Context, enabling technologies, and opportunities},
	author       = {Almasan, Paul and Ferriol-Galm{\'e}s, Miquel and Paillisse, Jordi and Su{\'a}rez-Varela, Jos{\'e} and Perino, Diego and L{\'o}pez, Diego and Perales, Antonio Agustin Pastor and Harvey, Paul and Ciavaglia, Laurent and Wong, Leon and others},
	year         = 2022,
	booktitle      = {IEEE CommMag}
}

@inproceedings{digitaltwin5g,
	title        = {Digital twin for {5G} and beyond},
	author       = {Nguyen, Huan X and Trestian, Ramona and To, Duc and Tatipamula, Mallik},
	year         = 2021,
	booktitle      = {IEEE CommMag},
}

@inproceedings{connectingtwins,
	title        = {Connecting the twins: A review on digital twin technology \& its networking requirements},
	author       = {Mashaly, Maggie},
	year         = 2021,
	booktitle      = {Procedia Computer Science}
}

@inproceedings{graphneuralnetwork,
	title        = {A graph neural network-based digital twin for network slicing management},
	author       = {Wang, Haozhe and Wu, Yulei and Min, Geyong and Miao, Wang},
	year         = 2020,
	booktitle      = {IEEE TII}
}

@inproceedings{intelligentdigitaltwin,
	title        = {Intelligent digital twin-based software-defined vehicular networks},
	author       = {Zhao, Liang and Han, Guangjie and Li, Zhuhui and Shu, Lei},
	year         = 2020,
	booktitle      = {IEEE Network}
}

@inproceedings{dwork2014algorithmic,
  title={The algorithmic foundations of differential privacy},
  author={Dwork, Cynthia and Roth, Aaron and others},
  booktitle={Found Trends Theor Comput Sci},
  year={2014}
}

@misc{udot2025traffic,
  author       = {{Utah Department of Transportation}},
  title        = {Traffic Data},
  howpublished = {\url{https://www.udot.utah.gov/connect/business/traffic-data/}},
  year         = {2025},
  note         = {Accessed: March 20th, 2025}
}

@inproceedings{zhang2025synergizing,
  title={Synergizing AI and Digital Twins for Next-Generation Network Optimization, Forecasting, and Security},
  author={Zhang, Zifan and Fang, Minghong and Chen, Dianwei and Yang, Xianfeng and Liu, Yuchen},
  booktitle={IEEE WCM},
  year={2025}
}

@article{he2022resource,
  title={Resource allocation based on digital twin-enabled federated learning framework in heterogeneous cellular network},
  author={He, Yejun and Yang, Mengna and He, Zhou and Guizani, Mohsen},
  journal={IEEE TVT},
  year={2022}
}

@article{lin20236g,
  title={6G digital twin networks: From theory to practice},
  author={Lin, Xingqin and Kundu, Lopamudra and Dick, Chris and Obiodu, Emeka and Mostak, Todd and Flaxman, Mike},
  journal={IEEE CommMag},
  year={2023}
}

@article{chen2024distributed,
  title={Distributed digital twin migration in multi-tier computing systems},
  author={Chen, Zhixiong and Yi, Wenqiang and Nallanathan, Arumugam and Chambers, Jonathon A},
  journal={IEEE JSTSP},
  year={2024}
}

@article{wang2022mobility,
  title={Mobility digital twin: Concept, architecture, case study, and future challenges},
  author={Wang, Ziran and Gupta, Rohit and Han, Kyungtae and Wang, Haoxin and Ganlath, Akila and Ammar, Nejib and Tiwari, Prashant},
  journal={IEEE IoT-J},
  year={2022}
}

@article{zhang2021digital,
  title={Digital twin empowered content caching in social-aware vehicular edge networks},
  author={Zhang, Ke and Cao, Jiayu and Maharjan, Sabita and Zhang, Yan},
  journal={IEEE TCSS},
  year={2021}
}

@article{zhang2025two,
  title={Two-phase authentication for secure vehicular digital twin communications},
  author={Zhang, Xinwei and Lai, Chengzhe and Li, Guanjie and Zheng, Dong},
  journal={Computer Networks},
  year={2025}
}

@article{wang2024secure,
  title={Secure and flexible data sharing with dual privacy protection in vehicular digital twin networks},
  author={Wang, Chenhao and Ming, Yang and Liu, Hang and Feng, Jie and Zhang, Ning},
  journal={IEEE TITS},
  year={2024}
}

@article{wang2025rsaka,
  title={RSAKA-VDT: Designing Reliable and Provably Secure Authenticated Key Agreement Scheme for Vehicular Digital Twin Networks},
  author={Wang, Kai and Dong, Jiankuo and Wang, Shiqin and Yuan, Zhijian and Sha, Letian and Xiao, Fu},
  journal={IEEE TVT},
  year={2025}
}

@article{jabeen2025securing,
  title={Securing Vehicle-to-Digital Twin Communications in the Internet of Vehicles},
  author={Jabeen Siddiqi, Sadia and Alobaidi, Abdulraheem H and Ahmad Jan, Mian and Tariq, Muhammad},
  journal={ACM TOMM},
  year={2025}
}

@article{li2024secr,
  title={SECR: A Secure and Efficient Charging Reservation Scheme Based on Digital Twin in Vehicular Network},
  author={Li, Guanjie and Luan, Tom H and Zheng, Jinkai and Lai, Chengzhe and Zhang, Kuan and Yu, Shui},
  journal={IEEE IoT-J},
  year={2024}
}

@article{ding2025digital,
  title={Digital-Twin-Enabled Federated Learning and CNN-based Channel Estimation for Urban Vehicular Channels},
  author={Ding, Cao and Ho, Ivan Wang-Hei},
  journal={IEEE IoT-J},
  year={2025}
}

@article{tang2025research,
  title={Research on Integrated Sensing, Communication Resource Allocation and Digital Twin Placement Based on Digital Twin in IoV},
  author={Tang, Lun and Wang, Asha and Xia, Bingsen and Tang, Yuanchun and Chen, Qianbin},
  journal={IEEE IoT-J},
  year={2025}
}

@article{tao2024ran,
  title={O-RAN-based digital twin function virtualization for sustainable IoV service response: An asynchronous hierarchical reinforcement learning approach},
  author={Tao, Yihang and Wu, Jun and Pan, Qianqian and Bashir, Ali Kashif and Omar, Marwan},
  journal={IEEE TGCN},
  year={2024}
}

@article{li2024intelligent,
  title={Intelligent reflecting surface and network slicing assisted vehicle digital twin update},
  author={Li, Li and Tang, Lun and Wang, Yaqing and Liu, Tong and Chen, Qianbin},
  journal={IEEE TITS},
  year={2024}
}

@article{tan2025energy,
  title={Energy-Efficient Federated Learning Training Optimization for Digital Twin Driven 6G Air-Ground Integrated Vehicular Networks},
  author={Tan, Can and Yu, Peng and Qu, Zhaowei and Zhang, Lixin and Li, Wenjing and Qiu, Xuesong and Guo, Shaoyong},
  journal={IEEE TITS},
  year={2025}
}

@article{li2025towards,
  title={Towards Intelligent Transportation with Pedestrians and Vehicles In-the-Loop: A Surveillance Video-Assisted Federated Digital Twin Framework},
  author={Li, Xiaolong and Wei, Jianhao and Wang, Haidong and Dong, Li and Chen, Ruoyang and Yi, Changyan and Cai, Jun and Niyato, Dusit and Shen, Xuemin},
  journal={IEEE Network},
  year={2025},
  publisher={IEEE}
}

@article{pegurri2025van3twin,
  title={VaN3Twin: the Multi-Technology V2X Digital Twin with Ray-Tracing in the Loop},
  author={Pegurri, Roberto and Gasco, Diego and Linsalata, Francesco and Rapelli, Marco and Moro, Eugenio and Raviglione, Francesco and Casetti, Claudio},
  journal={arXiv preprint arXiv:2505.14184},
  year={2025}
}

\end{document}